\documentclass[12pt, draftclsnofoot, onecolumn]{IEEEtran}

\usepackage{enumitem}

\usepackage[dvipsnames]{xcolor}

\usepackage{amssymb}

\usepackage[most]{tcolorbox}
\tcbset{on line, 
	boxsep=4pt, left=0pt,right=0pt,top=0pt,bottom=0pt,
	colframe=white,colback=gray!10,  
	highlight math style={enhanced}
}

\usepackage[noadjust]{cite}
\usepackage{filecontents}

\let\svtikzpicture\tikzpicture
\def\tikzpicture{\noindent\svtikzpicture}

\usepackage{amsbsy}
\usepackage{bm}
\usepackage{fixmath}
\usepackage{silence}
\usepackage{ctable}
\usepackage{pst-node,pst-plot}
\usetikzlibrary{shapes.geometric}
\usepackage{mathdots}
\usetikzlibrary{math}

\usepackage{upgreek}

\usepackage{autobreak}
\allowdisplaybreaks

\usepackage[caption=false]{subfig}
\usepackage[thinc]{esdiff}
\usepackage{commath,amsmath,amssymb,amsfonts}
\usepackage{mathtools}
\usepackage{amsthm}
\usepackage{algorithmic}
\usepackage{pgfplots} 
\pgfplotsset{compat=1.15}
\usepackage{graphicx}
\usepackage{pgfgantt}
\usepackage{pdflscape}
\usepackage{pst-plot}
\usepackage{xfrac}
\usepackage{colortbl}
\usepackage{cancel}
\usepgfplotslibrary{fillbetween}
\usepackage{amssymb,bm}
\usepackage{float}
\usepackage{amsmath}
\usepackage[draft=false]{hyperref}
\usepackage{multirow}

\usepackage{mathrsfs}
\usepackage{bbm}

\usepackage{cite} 

\usepackage{comment}

\usepackage{autobreak}
\allowdisplaybreaks

\def \D{\mathcal{D}}

\def \S{\mathcal{S}}

\def \fa{\mathbf{a}}
\def \fA{\mathbf{A}}
\def \fB{\mathbf{B}}

\def \fU{\mathbf{U}}
\def \fV{\mathbf{V}}

\def \fb{\mathbf{b}}
\def \fc{\mathbf{c}}

\def \fC{\mathbf{C}}
\def \fd{\mathbf{d}}

\def \fm{\mathbf{m}}

\def \fH{\mathbf{H}}
\def \fI{\mathbf{I}}

\def \fV{\mathbf{V}}
\def \fx{\mathbf{x}}
\def \fy{\mathbf{y}}
\def \fz{\mathbf{z}}

\def \fU{\mathbf{U}}

\def \fP{\mathbf{P}}

\def \fV{\mathbf{V}}
\def \fX{\mathbf{X}}
\def \fY{\mathbf{Y}}
\def \fZ{\mathbf{Z}}
\def \f0{\mathbf{0}}

\definecolor{blau_1b}{RGB}{0,90,169}
\definecolor{blau_2b}{RGB}{0,131,204}

\definecolor{mycolor7}{RGB}{152,205,225}
\definecolor{mycolor8}{RGB}{242,204,48}
\definecolor{mycolor9}{rgb}{0,.5,0}
\definecolor{mycolor10}{rgb}{.59,.44,.09}
\definecolor{mycolor11}{RGB}{231,199,31} 
\definecolor{mycolor12}{RGB}{8,133,161} 

\definecolor{cambridgeblue}{rgb}{0.64, 0.76, 0.68}
\definecolor{darkturquoise}{rgb}{0.0, 0.81, 0.82}
\definecolor{turquoiseblue}{rgb}{0.0, 1.0, 0.94}
\definecolor{turquoisegreen}{rgb}{0.63, 0.84, 0.71}
\definecolor{turquoise}{rgb}{0.19, 0.84, 0.78}
\definecolor{verdigris}{rgb}{0.26, 0.7, 0.68}
\definecolor{wheat}{rgb}{0.96, 0.87, 0.7}

\usepackage{accents}

\theoremstyle{remark} \newtheorem{theorem}{Theorem}
\theoremstyle{remark} 
\theoremstyle{remark} \newtheorem{definition}{Definition}
\theoremstyle{remark} 
\theoremstyle{remark} 
\theoremstyle{remark} \newtheorem{conjecture}{Conjecture}

\newtheorem{innercustomgeneric}{\customgenericname}
\providecommand{\customgenericname}{}
\newcommand{\newcustomtheorem}[2]{%
	\newenvironment{#1}[1]
	{%
		\renewcommand\customgenericname{#2}%
		\renewcommand\theinnercustomgeneric{##1}%
		\innercustomgeneric
	}
	{\endinnercustomgeneric}
}
\newcustomtheorem{customcorollary}{Corollary}
\newcustomtheorem{customlemma}{Lemma}
\newcustomtheorem{customproposition}{Proposition}

\usepackage{amssymb}

\usepackage{autobreak}
\allowdisplaybreaks

\usepackage{tikz}
\usetikzlibrary{arrows.meta}

\usepackage{blochsphere}

\usepackage{multicol, blindtext}

\usepackage[]{pgfplots}
\usepackage{pgfplotstable}
\usepgfplotslibrary{dateplot}
\pgfplotsset{/pgf/number format/use comma,compat=newest}

\usepackage{amsbsy}
\usepackage{bm}
\usepackage{fixmath}
\usepackage{silence}
\usepackage{ctable}
\usepackage{pst-node,pst-plot}
\usetikzlibrary{shapes.geometric}
\usepackage{mathdots}
\usetikzlibrary{math}

\usepackage{feyn}
\usepackage{mathtools,multirow}

\def\fc{{\bf c}}
\def\fA{{\bf A}}
\def\fU{{\bf U}}
\def\fV{{\bf V}}

\def\fB{{\bf B}}
\def\fx{{\bf x}}

\def\fSigma{\bm{\mathrm{\Sigma}}}

\usetikzlibrary{arrows, arrows.meta}
\usepackage{pgfplots}
\pgfplotsset{compat=1.15}
\usepgfplotslibrary{polar}

\usepackage[bottom]{footmisc}

\usepackage{cleveref}
\crefname{equation}{Eq}{Eqs}
\crefformat{equation}{[#2#1#3]}
\crefrangelabelformat{equation}{[#3#1#4]--[#5#2#6]}

\usepackage{tcolorbox}

\usepackage{lipsum}   
\usepackage{framed}
\usepackage{tikz}
\usetikzlibrary{decorations.pathmorphing,calc}
\pgfmathsetseed{1} 
\pgfdeclarelayer{background}
\pgfsetlayers{background,main}

\def \bG {\mathcal{G}_{\fH}}

\usepackage{nccmath}

\usepackage[svgnames]{xcolor}
\usepackage{tikz}
\usepackage{mathtools}

\def\nudge{.5}

\tikzset{axis/.style={ultra thick, Red!75!black, -latex, shorten <=-\nudge cm, shorten >=-2*\nudge cm}}
\tikzset{line/.style={thick,Green}}

\allowdisplaybreaks

\def\fc{{\bf c}}
\def\fA{{\bf A}}
\def\fU{{\bf U}}
\def\fV{{\bf V}}

\def\fB{{\bf B}}
\def\fx{{\bf x}}

\def\fSigma{\boldsymbol{\mathbf{\Sigma}}}

\makeatletter
\patchcmd{\@maketitle}
{\addvspace{0.5\baselineskip}\egroup}
{\addvspace{-.5\baselineskip}\egroup}
{}
{}

\usepackage[letterpaper,bindingoffset=0.2in,
left=.7in,right=.7in,top=.9in,bottom=1.1in,
footskip=.2in]{geometry}

\usepackage{xpatch}
\xapptocmd{\appendix}{%
}{}{\PatchFailed}

\begin{document}
 
 \title{Gaussian Identification in MIMO Channels}
 
 \author{
 	\vspace{0.5cm}
 	\fontsize{13}{13} \selectfont \IEEEauthorblockN{Mohammad Javad Salariseddigh\IEEEauthorrefmark{1} and Christian Deppe\IEEEauthorrefmark{2}
 	}
 	\\
 	\vspace{0.8cm} \fontsize{12}{12} \selectfont
	\IEEEauthorrefmark{1}Technical University of Darmstadt\\
	\IEEEauthorrefmark{2}Technical University of Braunschweig\vspace{8mm}
	\\
	\fontsize{12}{12} \selectfont \itshape
	\textcolor{verdigris!80!gray}{Unto Ibn Sina (Avicenna),\\ that eternal flame of intellect, whose light outlasts the stars (c.~980--1037).} \vspace{-0mm}
 	}

\maketitle

	\begin{abstract}
		We study the identification capacity of discrete-time Gaussian multiple-input multiple-output (MIMO) channels under per-antenna, per-symbol peak-amplitude constraints. We establish upper and lower bounds on the identification capacity that scale linearly with the channel rank $r.$ The number of reliably identifiable messages can grow super-exponentially with blocklength $n,$ following $M_n(R)=2^{(n\log n)R},$ where $R$ is an achievable rate. The channel rank provides a spatial-multiplexing advantage, enabling higher achievable rates and, consequently, stronger super-exponential codebook growth. These results characterize the rank-dependent gain in achievable identification performance.
	\end{abstract}

	\section{Introduction}
	
		The fundamental limits of communication depend not only on the statistical properties of the underlying channel but also on the objective imposed on the communication system. While classical channel coding concerns the reliable reconstruction of transmitted messages \cite{S48,Shannon1957}, identification coding addresses the different task of determining whether a particular message was transmitted \cite{AD89,AH08G,Ahlswede06gen}. This distinction leads to markedly different codebook-growth laws. In particular, with randomized encoders (RI) \cite{AD89}, the number of reliably identifiable messages can grow doubly exponentially with the blocklength \(n\), according to \(M_n(R) = 2^{2^{nR}}\). Deterministic identification (DI) \cite{AN99,Salariseddigh_IT} exhibits a fundamentally different codebook-growth regime from its randomized counterpart. In particular, continuous-alphabet channels \cite{Salariseddigh_IT,Colomer_2025} follow super-exponential growth. Specifically, for a broad class of channel models, including Gaussian channels with fading \cite{Salariseddigh_IT,Vorobyev25,Yuan22}, Poisson channels with and without memory and their affine extensions \cite{Salariseddigh-TMBMC,Salariseddigh_OJCOMS_23,Salariseddigh25_ITW}, the number of identifiable messages scales as \(M_n(R) = 2^{(n\log n)R}\). These results demonstrate that the fundamental limits of identification, particularly the growth of the codebook with blocklength, depend on the underlying channel model and the associated channel characteristics. This observation motivates the study of identification over multiple-input multiple-output (MIMO) channels and raises the question of how their spatial structure, channel gains, and fading characteristics influence the fundamental limits of identification and the design of identification codes. This remarkable scaling has motivated renewed interest in identification within post-Shannon theory, goal-oriented communication paradigms and semantic communication frameworks \cite{Goldreich12,Salariseddigh23_BSC_Future_Internet}, as well as in the development of explicit identification code constructions \cite{KT99,Verdu02,Zinoghli24,Vorobyev25,Lengerke25}.
		
		The development of MIMO communication was driven by the growing demand for data under limited spectrum and transmit-power resources. Unlike conventional approaches such as increasing bandwidth, using higher-order modulation, densifying networks, or raising transmit power, MIMO improves capacity by exploiting the spatial dimension of the wireless channel \cite{telatar1999capacity,foschini1996layered}. For an $N_T \times N_R$ MIMO link, the number of independent spatial data streams is limited by the channel rank: $$r = \operatorname{rank}(\mathbf{H}) \leq \min(N_T,N_R).$$ Thus, a full-rank channel can provide up to $\min(N_T,N_R)$ spatial degrees of freedom, although antenna correlation, limited angular spread, line-of-sight propagation, and other channel conditions may reduce this number \cite{telatar1999capacity}. MIMO communication employs multiple antennas at the transmitter, receiver, or both to exploit the spatial structure of a wireless propagation channel \cite{Marzetta10,TseViswanath2005,biglieri2007mimo}. Unlike a single-input single-output (SISO) system, which is described by a single complex channel coefficient, a MIMO system is characterized by a channel matrix representing the propagation paths between transmit--receive antenna pairs. This spatial structure enables two complementary techniques. \emph{Spatial diversity} transmits redundant signal representations across multiple antennas or paths, improving robustness against fading. \emph{Spatial multiplexing}, in contrast, transmits independent data streams simultaneously over the same time--frequency resource. When the channel matrix has sufficient rank and the receiver has adequate channel-state information and processing capability, these streams can be separated to increase spectral efficiency \cite{FoschiniGans1998,telatar1999capacity,ZhengTse2003}. Massive MIMO systems consider wireless links equipped with a very large number of transmit and/or receive antennas \cite{Hoydis11,Hoydis13,Wang24} and therefore provides a natural framework for identification, particularly in rich-scattering scenarios such as i.i.d. Rayleigh fading \cite{marzetta1999capacity}, where the channel matrix is full rank with probability one and the resulting spatial degrees of freedom can dramatically increase the identification codebook size.
		
		The Gaussian MIMO channel constitutes a canonical model for modern wireless communication systems \cite{Proakis01,G05}. In particular, the Gaussian model isolates the role of the channel matrix and its spatial degrees of freedom, making channel rank a natural parameter for studying the impact of spatial dimensions on communication. For the classical Shannon transmission problem, consider a point-to-point Gaussian MIMO channel with channel matrix $\mathbf H$. At high SNR, its capacity satisfies \cite{telatar1999capacity} $\mathbb C = r \log(\mathrm{SNR})+O(1).$ Thus, the rank determines the number of independent spatial dimensions and the corresponding spatial multiplexing gain. At the first-order scale, the logarithm of the Shannon codebook size satisfies $\log M_n = n\mathbb C + o(n) .$ Thereby, using Telatar's high-SNR characterization, $\log M_n = nr\log(\mathrm{SNR})+O(n).$ Thus, the channel rank appears directly as the coefficient of the leading-order growth of the logarithm of the achievable codebook size. This raises the question of whether channel rank plays an analogous role in identification. Despite the fundamentally different operational objectives of identification and message transmission, our results show that the achievable identification rates scale also linearly with the channel rank. Thus, although spatial multiplexing has different operational interpretations in the two settings, the channel rank $r$ plays an analogous structural role: it determines the multiplicative factor in the leading-order growth of the logarithm of the achievable message-set size. Relative to a rank-one channel, this yields an asymptotic $r$-fold spatial gain, with the growth scaling as $n\log(\mathrm{SNR})$ for Shannon transmission and $n\log n$ for Ahlswede identification problem.

		A related line of research on RI over Gaussian MIMO channels has been pursued in \cite{LabidiDeppeBoche2021,LabidiDeppeBoche2020,LabidiDeppeBoche2025}. However, to the best of our knowledge, DI over MIMO Gaussian channels has not yet been studied. In contrast, identification capacity has been characterized for SISO Gaussian channels under both white-noise and colored-noise assumptions \cite{Salariseddigh_IT,Salari_ISI,Salari_Colored}. Thus, the characterization of identification capacity for MIMO Gaussian channels under deterministic encoding remains an open problem, including the white-noise setting considered in this paper and extensions to intersymbol interference (ISI) and colored noise. In this paper, we address this gap by studying DI over static MIMO Gaussian channels with white noise and full channel state information (CSI) available at both the transmitter and receiver. We consider deterministic encoders subject to a peak-amplitude constraint and assume that the channel matrix remains constant throughout each transmission block. A key step in our rate analysis is the isometric equivalence between the Frobenius-normed matrix space and the Euclidean space
		$
		\bigl(\mathbb{R}^{r\times n},\|\cdot\|_F\bigr)
		\cong_{\mathrm{iso}}
		\bigl(\mathbb{R}^{rn},\|\cdot\|_2\bigr),
		$
		induced by the vectorization map
		\(\operatorname{vec}:\mathbb{R}^{r\times n}\to\mathbb{R}^{rn}\).
		This observation allows us to transfer the relevant Euclidean-space arguments directly to the matrix setting while preserving distances and norms. In particular, this isometry plays a crucial role in deriving the stated convergence-rate bounds without introducing additional constants arising from norm equivalence. The lower bound follows from a structured DI code that exploits the directional gains of the channel and high-dimensinoal sphere packing within the hypercube, while the upper bound follows from a converse that translates the identfication error requirements into pairwise separation conditions and leverage volume argument in the induced output space. These bounds quantify how the DI capacity scales with the rank and singular values of the channel matrix.

	\subsection{Related Works in Shannon Problem}

	The information-theoretic foundations of Gaussian MIMO communication date back to Shannon's formulation of channel capacity, where reliable communication rates are characterized by maximizing mutual information over admissible input distributions~\cite{S48}. The introduction of multiple antennas in the 1990s added spatial processing as a new degree of freedom for wireless communication. Foschini proposed layered space--time architectures to exploit spatial multiplexing in fading channels~\cite{foschini1996layered}, while Foschini and Gans established the capacity scaling of multiple-element antenna systems under fading~\cite{FoschiniGans1998}. Telatar later provided the seminal characterization of point-to-point Gaussian MIMO capacity for both deterministic and fading channels~\cite{telatar1999capacity}. Consider the deterministic MIMO channel
	\begin{equation}
		\mathbf{Y}=\mathbf{H}\mathbf{X}+\mathbf{Z},
	\end{equation}
	where $\mathbf{H}\in\mathbb{C}^{n_r\times n_t}$ is the channel matrix and
	$\mathbf{Z}\sim\mathcal{CN}(\mathbf{0},\sigma^2\mathbf{I}_{n_r})$ is additive Gaussian noise. Under an input covariance constraint
	$\mathbb{E}\!\left[\mathbf{X}\mathbf{X}^{\mathsf H}\right]\preceq\mathbf{Q}$, the capacity under a total average-power constraint is
	\begin{equation}
		\mathbb C(\mathbf{H},P_{\mathrm a})
		=
		\max_{\substack{\mathbf{Q}\succeq 0\\
		\operatorname{tr}(\mathbf{Q})\leq P_{\mathrm a}}}
		\log \det\!\left(
		\mathbf{I}_{n_r}
		+
		\frac{1}{\sigma^2}
		\mathbf{H}\mathbf{Q}\mathbf{H}^{\mathsf H}
		\right),
		\label{eq:mimo_capacity_average}
	\end{equation}
	where $\mathbf{Q}$ is the input covariance matrix. The capacity-achieving input is circularly symmetric complex Gaussian, with the available power allocated across the eigenmodes of $\mathbf{H}^{\mathsf H}\mathbf{H}$ according to the classical water-filling solution~\cite{telatar1999capacity}. For fading channels, the capacity further depends on the channel model and on the available channel-state information~\cite{FoschiniGans1998,telatar1999capacity,marzetta1999capacity}. The average-power formulation, however, does not capture systems in which the transmitted signal is subject to a strict instantaneous power limit. In addition to the average-power constraint $\mathbb{E}\!\left[\|\mathbf{X}\|^2\right]\leq P_{\mathrm a},$ one may instead impose a peak-power constraint $\|\mathbf{X}\|^2\leq P_{\mathrm p},$ for every channel use, where $P_{\mathrm a}\leq P_{\mathrm p}$. This constraint changes the nature of the capacity-achieving input distribution. In particular, a Gaussian input is no longer feasible under a strict peak-power constraint because its support is unbounded. The capacity problem must therefore be optimized over input distributions satisfying the peak-power constraint, rather than being characterized solely through the input covariance matrix.
	
	The fundamental distinction between average-power and peak-constrained
	Gaussian channels was established in the scalar setting by Smith. For the
	real AWGN channel with a hard amplitude constraint, Smith showed that the
	capacity-achieving input distribution is discrete with a finite number of
	mass points~\cite{Smith71}. Shamai and Bar-David extended this
	result to the complex Gaussian channel with simultaneous average- and
	peak-power constraints, showing that the optimal distribution is discrete in
	amplitude with a finite number of mass points and has a uniformly distributed
	phase~\cite{shamai1995capacity}. Thus, imposing a hard instantaneous
	constraint can fundamentally change the structure of the optimal input
	distribution, even though the underlying channel remains Gaussian.
	
	The corresponding vector and MIMO problem was studied by Rassouli and
	Clerckx under simultaneous average- and peak-power constraints
	~\cite{rassouli2016capacity}. For the isotropic identity channel, they
	established finite radial support of the capacity-achieving distribution,
	corresponding to a finite collection of hyperspheres in the input space. For
	general channel matrices, they derived capacity bounds and structural
	conditions characterizing the optimal input distribution
	~\cite{rassouli2016capacity}.
	
	Hence, Gaussian MIMO channels exhibit two fundamentally different capacity
	regimes. Under an average-power constraint alone, Gaussian signaling is
	optimal and the capacity is characterized by Telatar's covariance-based
	formulation in \eqref{eq:mimo_capacity_average}.
	Under a hard peak-power constraint, Gaussian signaling is generally infeasible,
	and capacity must instead be optimized over bounded-support input
	distributions. The resulting optimal distribution can exhibit finite discrete
	structure in amplitude or radius, depending on the channel geometry and the
	relationship between the average and peak powers
	~\cite{Smith71,shamai1995capacity,rassouli2016capacity}.
	This distinction motivates analytical characterizations and computable bounds
	for peak-constrained Gaussian MIMO channels, particularly when instantaneous
	transmit-power limitations make the conventional covariance-based Gaussian
	solution insufficient.

	\subsection{Notations}

	Throughout, $\mathbbmss{X},\mathbbmss{Y},\mathbbmss{Z},\ldots$ denote alphabet sets, and $\mathbbmss{X}^{c}$ denotes the complement of $\mathbbmss{X}$. Lowercase italic letters $x,y,z,\ldots$ represent scalar constants or realizations of scalar random variables (RVs), whereas uppercase italic letters $X,Y,Z,\ldots$ denote scalar RVs. Lowercase bold symbols $\fx,\fy,\fz,\ldots$ represent vector or matrix realizations, and uppercase bold symbols $\fX,\fY,\fZ,\ldots$ denote random vectors or matrices. We use $[\![M]\!]\triangleq\{1,2,\ldots,M\}$ to denote the set of integers from $1$ to $M$, $\mathbb{R}_{+}$ and $\mathbb{R}$ to denote the sets of non-negative real numbers and real numbers, respectively, and $\triangleq$ to indicate equality by definition. All logarithms are taken to base $2$. For a positive integer $x$, the factorial is defined as
	\begin{equation}
		x!\triangleq\prod_{k=1}^{x}k,
	\end{equation}
	and the Gamma function is given, for positive integers $x$, by
	\begin{equation}
		\Gamma(x)=(x-1)!.
	\end{equation}
	We write $X\sim\mathcal{N}$ when $X$ follows the distribution $\mathcal{N}$, and use the standard Bachmann--Landau notation $o(\cdot)$, $O(\cdot)$, and $\sim$. The expectation and variance of an RV $X$ are denoted by
	\begin{align}
		\mathbb{E}[X] &\quad\text{and}\quad
		\operatorname{Var}[X]
		\triangleq \mathbb{E}\!\left[
		(X-\mathbb{E}[X])^2
		\right],
	\end{align}
	respectively. The quantities $\|\mathbf{x}\|_1$, $\|\mathbf{x}\|$, $\|\mathbf{x}\|_{\max}$, $\|\fX\|_{\max}$, and $\|\fX\|_F$ denote the $\ell_1$-norm, Euclidean norm, maximum-entry norm of a vector, maximum-entry norm of a matrix, and Frobenius norm of a matrix, respectively. For $\fx_0\in\mathbb{R}^n$ and $\fX_0\in\mathbb{R}^{m\times n}$, we define the corresponding Euclidean and Frobenius-norm balls as
	\begin{align}
		\mathcal{S}_{\fx_0}(n,\tau)
		&\triangleq
		\bigl\{
		\fx\in\mathbb{R}^n:
		\|\fx-\fx_0\|\leq\tau
		\bigr\},\\
		\mathcal{S}_{\fx_0}(m,n,\tau)
		&\triangleq
		\bigl\{
		\fx\in\mathbb{R}^{m\times n}:
		\|\fx-\fx_0\|_F\leq\tau
		\bigr\}.
	\end{align}
	The corresponding hypercubes are defined as
	\begin{align}
		\mathbbmss{Q}_{\mathbf{0}}(n,2U)
		&\triangleq
		\bigl\{
		\fx\in\mathbb{R}^n:
		|x_t|\leq U,\ 
		\forall t\in[\![n]\!]
		\bigr\},\\
		\mathbbmss{Q}_{\mathbf{0}}(m,n,2U)
		&\triangleq
		\bigl\{
		\fx\in\mathbb{R}^{m\times n}:
		|x_{tt'}|\leq U,\ 
		\forall t\in[\![m]\!],\
		t'\in[\![n]\!]
		\bigr\}.
	\end{align}
	For continuous probability distributions $P$ and $Q$ on $\mathbb{R}^n$, with respective densities $p$ and $q$, their fidelity is defined as
	\begin{equation}
		F(P,Q)\triangleq
		\int_{\mathbb{R}^n}
		\sqrt{p(\fy)q(\fy)}\,d\fy.
	\end{equation}
	The notation $|\cdot|$ denotes the absolute value of a scalar or the determinant of a matrix, as determined by the context. The Frobenius inner product of $\fX,\fY\in\mathbb{R}^{m\times n}$ is defined as
	\begin{equation}
		\langle\fX,\fY\rangle_F
		\triangleq
		\operatorname{tr}(\fX^{\mathsf T}\fY).
	\end{equation}
	Furthermore, $\operatorname{vec}(\cdot)$ denotes the vectorization operator that stacks the columns of a matrix into a single column vector. Finally, $\mathcal{G}_{\fH}$ denotes the discrete-time multiple-input multiple-output (MIMO) Gaussian channel characterized by the channel matrix $\fH$.

	\section{System Model}
	\label{Sec.SysModel}
	Here, we introduce the adopted system model and establish preliminaries for coding and capacity.

	\subsection{MIMO Gaussian Communication Model}
	
	We consider a discrete-time, real-baseband MIMO Gaussian communication system with $N_T$ transmit antennas and $N_R$ receive antennas where the propagation channel is characterized by a deterministic channel matrix $\fH = [h_{\nu\mu}] \in \mathbb{R}^{N_R\times N_T}$ which is assumed to remain constant throughout a transmission block of $n$ channel uses, and perfectly known at the transmitter and receiver referred to as the full channel state information at both ends. Since $\fH$ is deterministic and fixed over the transmission block, its singular values do not vary with the blocklength $n$; hence, for fixed $N_R$ and $N_T$, $\{\sigma_k(\fH)\}_{k=1}^r$ are finite constants with respect to $n$. In particular, $\sigma_{1}(\fH)<\infty,$	and, whenever $r>0$, it holds $0<\sigma_{r}(\fH)\leq\sigma_{1}(\fH)<\infty.$ Let
	\begin{equation}
		r \triangleq \operatorname{rank}(\fH) \leq \min\{N_R,N_T\}.
	\end{equation}
	The singular values of $\fH$ are ordered as
	\begin{equation}
		\sigma_1(\fH)\geq\sigma_2(\fH)\geq\cdots\geq\sigma_r(\fH)>0,
	\end{equation}
	while, if $\fH$ is rank-deficient, i.e., $r<\min\{N_R,N_T\}$, the remaining singular values satisfy
	\begin{equation}
		\sigma_{r+1}(\fH)=\cdots=\sigma_{\min\{N_R,N_T\}}(\fH)=0.
	\end{equation}
	Thus, the channel has $r$ nonzero spatial modes, with
	$\sigma_1(\fH)$ and $\sigma_r(\fH)$ representing, respectively, the
	strongest and weakest nonzero spatial gains. No assumption that all singular values are equal or that the channel is	full rank is imposed. In general, the nonzero singular values may beunequal, reflecting the different gains of the available spatial  modes, and the channel may contain a nontrivial null space corresponding to its zero singular values.

	Specifically, the $(\nu,\mu)$-th entry $h_{\nu\mu}$ of $\mathbf H$ denotes the
	channel gain from transmit antenna $\mu\in[\![N_T]\!]$ to receive antenna
	$\nu\in[\![N_R]\!]$	where $\nu\in[\![N_R]\!]$ indexes the receive antenna and
	$\mu\in[\![N_T]\!]$ indexes the transmit antenna. Thus,
	\begin{align*}
		\fH_{N_R \times N_T} =
		\begin{bmatrix}
			h_{11} & h_{12} & \cdots & h_{1N_T}\\
			h_{21} & h_{22} & \cdots & h_{2N_T}\\
			\vdots & \vdots & \ddots & \vdots\\
			h_{N_R1} & h_{N_R2} & \cdots & h_{N_RN_T}
		\end{bmatrix}
		\in \mathbb{R}^{N_R\times N_T} .
	\end{align*}
	The block-constant assumption means that the channel coefficients $H_{\nu\mu}$
	do not vary with the channel-use index $t\in[\![n]\!]$ within the transmission
	block. Equivalently, if $\mathbf H_t$ denotes the channel matrix at channel use
	$t$, then $\mathbf H_t=\mathbf H, \forall t\in[\![n]\!].$
	Hence, the same channel matrix governs the propagation of every transmitted vector
	throughout the block. This model corresponds to a channel whose coherence time is
	at least as long as the transmission interval under consideration; any channel
	variation, if present, is assumed to occur only between distinct transmission blocks
	and is therefore not resolved within a single block.

	\subsection{Channel Inputs}
	
	At channel use $t\in[\![n]\!]$, let $\fX_t \triangleq \big((X_{t,\mu})_{\mu=1}^{N_T}\big)^T\in\mathbb{R}^{N_T}$ denote the transmit vector, where $\mu\in[\![N_T]\!]$ indexes the transmit antenna and
	$X_{t,\mu}$ is the symbol transmitted from antenna $\mu$. Thus, $N_T$ symbols are transmitted simultaneously at each channel use. For a fixed antenna $\mu$, the
	corresponding length-$n$ transmit sequence is $\fX_\mu \triangleq (X_{t,\mu})_{t=1}^n\in\mathbb{R}^n.$ Hence, the complete channel input over $n$ channel uses is the random matrix $\fX=[X_{t,\mu}]\in\mathbb{R}^{N_T\times n},$ represented either column-wise by $\{\fX_t\}_{t=1}^n$ or row-wise by $\{\fX_\mu\}_{\mu=1}^{N_T}$. We impose a per-symbol peak-amplitude constraint on the channel input sequence. Specifically, every admissible realization $\fx$ is required to satisfy
	\begin{equation}
		 \| \fx \|_{\max} \le \sqrt{P_{\max}},
	\end{equation}
	 equivalently, $|x_{t,\mu}| \le \sqrt{P_{\max}}, \forall\,t\in[\![n]\!],\forall\,\mu\in[\![N_T]\!].$ Here, $P_{\max}>0$ denotes the maximum allowable transmit amplitude per antenna and channel use.

	\subsection{Channel and Noise Model}

	Let $\fX_t \in \mathbb{R}^{N_T}$, $\fY_t \in \mathbb{R}^{N_R}$, and $\fZ_t \in \mathbb{R}^{N_R}$ denote the random transmit vector, received vector, and additive noise vector, respectively, at channel use $t.$
	Specifically, during each channel use, the transmitter emits a vector of symbols across $N_T$ transmit antennas, in contrast to conventional SISO schemes where only a scalar symbol is transmitted. The input-output relation of the channel $\mathcal{G}_{\fH}$ at time $t \in [\![n]\!]$ is given by 
	\begin{equation}
		\fY_t = \fH \fX_t + \fZ_t .
	\end{equation}
	 Hence, at every channel use, the transmitter simultaneously sends $N_T$ symbols over the $N_T$ transmit antennas, and the receiver observes an $N_R$-dimensional vector corrupted by additive noise and shaped by the channel matrix $\fH.$ The standard transmission scheme employs input strings of length $n$. Hence, the block-wies channel law admits the following compact form:
	\begin{align}
		\label{Eq.Block_Wise_Channel_Law}
		\mathbf{Y} = \fH \fX + \mathbf{Z},
	\end{align}
	where $\fX \in \mathbb{R}^{N_T \times n},\ \fY \in \mathbb{R}^{N_R \times n},\ \fZ \in \mathbb{R}^{N_R \times n}$ are input, output and noise matrices. We assume that the noise features spatial (across antennas) and temporal (across channel uses) independence. Let \(Z_{t,\mu}\) denote the noise sample at receive antenna $\mu \in [\![N_R]\!]$ and time index $t \in [\![n]\!].$ The independent and identically distributed (i.i.d.) Gaussian assumption means 
	\begin{equation}
		Z_{t,\nu} \overset{\text{i.i.d.}}{\sim} \mathcal{N}(0,\sigma^2),
	\end{equation}
	 which is equivalently characterized by
	 \begin{equation}
	 	\mathbb{E}[Z_{t,\nu}\,Z_{t',\nu'}] = \sigma^2\,\delta_{\nu\nu'}\,\delta_{tt'},
	 \end{equation}
	  for $\nu,\nu' \in [\![N_R]\!]$ and $t,t' \in [\![n]\!]$ where \(\delta_{\nu \nu'}\) and \(\delta_{tt'}\) are Kronecker deltas. This captures independence across both receive antennas and time. Therefore, $\fZ_t \overset{\text{\tiny i.i.d.}}{\sim} \mathcal{N}(\mathbf{0}, \sigma^2 \mathbf{I}_{N_R})$ is the receiver noise vector at a single time instant $t$ with the following distribution
	\begin{align}
		\label{Eq.Noise_Density_Acr_N_R}
		f_{\fZ_t}(\fz_t) = (2\pi\sigma^2)^{-N_R n/2} \exp(- \| \fz_t \|^2 / 2\sigma^2 ) ,
	\end{align}
	with $\fz_t \in \mathbb R^{N_R}$. Each column $\mathbf z_t$ of the noise matrix $\fZ$ is i.i.d. according to
	\eqref{Eq.Noise_Density_Acr_N_R}. Since the noise is independent across	time and receive antennas, the joint block-wise density factorizes as
	\begin{equation}
		f_{\fZ}(\fz) = (2\pi\sigma^2)^{-N_R n/2} \exp( -\|\fz\|_F^2/2\sigma^2 ).
	\end{equation}
	Equivalently, $\fZ \sim \mathcal{N}(\mathbf 0,\sigma^2\mathbf I_{N_R n}).$ Conditioned on the transmitted vector $\mathbf{x}_t$, the received vector satisfies
	\begin{equation}
		\fY_t \mid \fX_t \sim \mathcal{N}(\fH \fX_t,\sigma^2\fI_{N_R}).
	\end{equation}
	 Assuming the block channel law given in \eqref{Eq.Block_Wise_Channel_Law}, the conditional density of the received matrix $\fY$ given $\fX$ and $\fH$ factorizes as
	\begin{align}
		\label{Eq.Noise_Den_Matr}
		f_{\fY|\fX}(\fy|\fx)
		= (2\pi\sigma^2)^{-N_R n/2} \cdot \exp\big(- \big\| \fy - \fH\fx \big\|_{F}^2 / 2\sigma^2 \big) .
	\end{align}

	\subsection{Rank-Reduced Channel Representation}
	
	Invoking the reduced SVD given in Definition \ref{Def.SVD} and setting $p=N_R$ and $q=N_T$, we observe that the rank constraint implies that only the $r$-dimensional right-singular subspace of $\fH$ contributes to the received signal. Accordingly, define the rank-space representation of the channel input as
	\begin{equation}
		\label{Eq.bar_X}
		\bar{\fX}_t \triangleq \fV_r \fX_t \in\mathbb{R}^{N_T \times n}.
	\end{equation}
	 Thereby, using the SVD of $\fH$ and $\fV_r \fV_r^{\mathsf T} = \fI$ gives
	\begin{align}
		\fY_t = (\fU_r\fSigma_r\fV_r^{\mathsf T} ) \bar{\fX} + \fZ_t = \fU_r\fSigma_r \fX_t + \fZ_t .
	\end{align}
	Next, we proceed to project the received signal and noise vectors onto the $r$-dimensional column space of $\fH.$ Specifically, define the projected signal and noise as $\bar{\fY}_t \triangleq \fU_r^{\mathsf T}\fY_t,$ and $\bar{\fZ}_t \triangleq \fU_r^{\mathsf T}\fZ_t.$ Thus, the channel $\forall t \in [\![n]\!]$ reduces to
	\begin{align}
		\bar{\fY}_t = \fU_r^{\mathsf T} \bar{\fY}_t = \fU_r^{\mathsf T} \fU_r\fSigma_r \fX_t + \fU_r^{\mathsf T} \fZ_t = \fSigma_r \fX_t + \bar{\fZ}_t .
	\end{align}
	Since $\bar{\fZ}_t=\fU_r^{\mathsf T}\fZ_t,\ \fU_r^{\mathsf T}\fU_r=\fI_r$, we have $\bar{\fZ}_t\sim\mathcal{N}(\mathbf{0},\sigma^2\fI_r).$ Thus, a rank-$r$ MIMO channel is equivalently represented by the $r$-dimensional channel matrix $\fSigma_r$. The remaining dimensions lie in the channel null spaces and do not affect the received signal; only the codeword projection onto the $r$-dimensional right-singular subspace matters.

	\section{Coding Scheme and Preliminaries}
	
	In the following, we draw on the rigorous performance parameters for identification problem including the formulation of the code definition and capacity for the MIMO Gaussian channel $\bG.$

		\begin{definition}[Singular Value Decomposition]
		\label{Def.SVD}
		Let $\fH\in\mathbb{R}^{p\times q}$ be a matrix of rank $r\triangleq\operatorname{rank}(\fH)\leq\min\{p,q\}.$
		The singular value decomposition (SVD) of $\fH$ is given by $\fH=\fU\fSigma\fV^{\mathsf T},$
		where $\fU\in\mathbb{R}^{p\times p}, \fV\in\mathbb{R}^{q\times q}, \fU^{\mathsf T}\fU=\fI_p, \fV^{\mathsf T}\fV=\fI_q,$ and $\fSigma\in\mathbb{R}^{p\times q}$ contains the singular values $\sigma_1\geq\cdots\geq\sigma_r>0,$ with all remaining singular values equal to zero. Equivalently, retaining only the $r$ nonzero singular values yields the \emph{reduced SVD} 
		\begin{equation}
			\fH=\fU_r\fSigma_r\fV_r^{\mathsf T},
		\end{equation}
		 where $\fU_r\in\mathbb{R}^{p\times r}, \fV_r\in\mathbb{R}^{q\times r}, \fU_r^{\mathsf T}\fU_r = \fV_r^{\mathsf T}\fV_r = \fI_r,$ and $\fSigma_r = \operatorname{diag}(\sigma_1,\ldots,\sigma_r) \in\mathbb{R}^{r\times r}.$
		\qed
	\end{definition}

	\begin{definition}[Identification Code]
		\label{Def.MIMO_ID_Code}
		
		An $(n,M_n(R),e_1,e_2)$-code for the MIMO channel $\bG$ under the
		peak-amplitude constraint $\sqrt{P_{\max}}$ is defined by the pair
		$(\mathbbmss{C},\mathcal{D})$, where $n$ denotes the blocklength,
		$R$ the coding rate, and $M_n(R)\in\mathbb{N}$ the number of messages.
		Specifically,
		\[
		\mathbbmss{C}
		=
		\{\fc_i\}_{i=1}^{M_n(R)},\quad
		\mathcal{D}
		=
		\{\mathbbmss{D}_i\}_{i=1}^{M_n(R)},
		\]
		where $\mathbbmss{C}$ is the codebook and $\mathcal{D}$ is the collection of decoding regions. Each codeword matrix $\mathbf{c}_i \in \mathbb{R}^{r\times n}$ is given by $\mathbf{c}_i=[c_{i,t,k}]
		\in\mathbb{R}^{r\times n},$
		where $i\in[\![M_n(R)]\!]$ denotes the message index,
		$t\in[\![n]\!]$ denotes the channel-use index, and
		$k\in[\![r]\!]$ denotes the spatial-stream index.
		Specifically, $c_{i,t,\nu}$ is the $(\nu,t)$-th entry of
		$\mathbf{c}_i$, that is, $[\mathbf{c}_i]_{\nu,t}=c_{i,t,\nu}.$ Let $\mathbf{V}_r\in\mathbb{R}^{N_T\times r}$ denote the corresponding
		precoding matrix. The matrix transmitted over the $N_T$ antennas is
		then
		\begin{align}
			\label{Eq.Precoder}
			\bar{\fc}_i
			\triangleq
			\fV_r \fc_i
			\in\mathbb{R}^{N_T\times n}.
		\end{align}
		Since the peak-amplitude constraint concerns the signal physically
		transmitted over the antennas, it is imposed on the entries of
		$\bar{\mathbf{c}}_i$. Thus, $\forall i\in[\![M_n(R)]\!]$,
		$t\in[\![n]\!]$, and
		$\mu\in[\![N_T]\!]$,
		\begin{equation}
			|\bar{c}_{i,t,\mu}|
		\leq
		\sqrt{P_{\max}},
		\end{equation}
		where $\bar{c}_{i,t,\mu}$ denotes the $(\mu,t)$-th entry of
		$\bar{\mathbf{c}}_i$, with $\mu$ indexing the physical transmit
		antennas. For a codeword $\fc_i$, the corresponding channel input is given by $\bar{\fc}_i=\fV_r \fc_i$. Hence, the conditional distribution of the channel output is determined by $\bar{\fc}_i$. Accordingly, two types of decoding errors may occur: a type-I error and a type-II error. The corresponding error probabilities read
		\begin{align}
					P_{e,1}(i) & \triangleq
					\Pr\big(\fY\in\mathbbmss{D}_i^c\,\big|\,\fX=\bar{\fc}_i\big)
					=
					1-\int_{\mathbbmss{D}_i}
					f_{\fZ}(\fy-\fH\bar{\fc}_i)\,d\fy .
					\\
					P_{e,2}(i,j) & \triangleq
					\Pr\big(\fY\in\mathbbmss{D}_j\,\big|\,\fX=\bar{\fc}_i\big)
					=
					\int_{\mathbbmss{D}_j}
					f_{\fZ}(\fy-\fH\bar{\fc}_i)\,d\fy . \
				\end{align}
		The code satisfies
		\begin{align}
			P_{e,1}(i) \leq e_1, \qquad 
			P_{e,2}(i,j) \leq e_2, \qquad 	\forall\, i,j\in[\![M_n(R)]\!], i \neq j ,
		\end{align}
		 for prescribed error tolerances $e_1,e_2>0$.
		\qed
	\end{definition}

	\begin{definition}[Identification Capacity] A rate $R>0$ is achievable if, for arbitrary $e_1,e_2>0$ and all sufficiently large blocklengths $n$, there is an $(n,M_n(R),e_1,e_2)$-code satisfying Definition~\ref{Def.MIMO_ID_Code}. The identification capacity of $\bG$ is supremum of all achievable rates and is denoted by $\mathbb{C}_{\text{I}}(\bG).$
	\qed
	\end{definition}

	\subsection{Operational Interpretation of Identification Rates}

	 For DI problem, the logarithm of the codebook
	 size can grow on the order of $n\log n$. Consequently, the
	 conventional Shannon coding rate, normalized by $n$, does not
	 adequately characterize the dominant asymptotic growth of the
	 codebook. This motivates a two-scale characterization of
	 identification performance. The leading $n\log n$ term determines the dominant growth of
	 the number of identification messages, while the subsequent
	 linear-in-$n$ term captures finer, system-dependent contributions.
	 Operationally, these two scales correspond to multiplicative
	 factors of the form $n^{nL_1}$ and
	 $e^{nL_2}$, respectively,
	 in an asymptotic expansion of the form
	 \begin{equation}
	 	M_n
	 =
	 n^{nL_1(r)}
	 \exp\!\left\{
	 nL_2(\sigma_r(\fH),r,P_{\max})+o(n)
	 \right\}.
	 \end{equation}
	 
	 Thus, $L_1$ characterizes the dominant codebook-growth scale,
	 whereas $L_2(\sigma_r(\fH),r,P_{\max})$ captures the
	 system-dependent contribution at the next asymptotic order.
	 The latter can distinguish coding schemes with the same
	 leading-order growth but different system-dependent performance,
	 provided the corresponding asymptotic expansion is established. The achievability and converse bounds establish
	 \begin{equation}
	 	\log M_n(R)\asymp n\log n,
	 	\qquad
	 	M_n(R)=\exp\bigl(\Theta(n\log n)\bigr)
	 	=n^{\Theta(n)},
	 \end{equation}

	 Throughout the analysis, we use $R$ to denote both the
	 finite-blocklength coding rate and the corresponding
	 asymptotically achievable rate, writing $R_n$ when its
	 dependence on $n$ needs to be emphasized. This convention
	 simplifies the transition between finite-blocklength bounds,
	 asymptotic achievability, and capacity statements. For a
	 detailed discussion of this distinction and the associated
	 asymptotic arguments, we refer the reader to
	 \cite{Salari26_Inverse}.

	 	\begin{definition}[Coding Rates]
	 	\label{Def.Coding_Rate}
	 	For an $(n,M_n,e_1,e_2)$-code, the classical Shannon coding
	 	rate reads
	 	\begin{align}
	 		R_n^{\mathrm S}
	 		\triangleq \frac{\log M_n}{n}.
	 		\label{Eq.S_Coding_Rate}
	 	\end{align}
	 	For DI, we define the first-order
	 	identification rate as
	 	\begin{align}
	 		R_1
	 		\triangleq
	 		\liminf_{n\to\infty}
	 		\frac{\log M_n}{n\log n}.
	 		\label{Eq.First_Order_Rate}
	 	\end{align}
	 	This rate characterizes the dominant asymptotic growth of
	 	the logarithm of the codebook size.
	 	
	 	Let $L_1 \triangleq (1-b)r/4$ denote the limiting first-order
	 	coefficient established by the first-order analysis. The
	 	second-order identification rate, relative to this coefficient,
	 	is defined as
	 	\begin{align}
	 		R_2
	 		\triangleq
	 		\liminf_{n\to\infty}
	 		\frac{\log M_n-L_1 n\log n}{n}.
	 		\label{Eq.Second_Order_Rate}
	 	\end{align}
	 \end{definition}

	\subsection{Encoding}
	
	We consider a deterministic encoder defined as
	\begin{equation}
		f:[\![M_n(R)]\!]\to\mathbb{R}^{r \times n},
		\qquad
		i\mapsto\fc_i,
	\end{equation}
	where each message $i\in[\![M_n(R)]\!]$ is uniquely mapped to a matrix-valued codeword
	$\fc_i\in\mathbb{R}^{r \times n}$. Upon observing message $i$, the transmitter sends the
	corresponding codeword $\fc_i$. If the transmitted message is modeled as a random variable
	$I$, then the resulting transmitted codeword $\fC=f(I)$ is a random matrix induced by the deterministic encoder. The full transmitted codeword $\fc_i\in\mathbb{R}^{r \times n}$ associated with message $i \in [\![M_n(R)]\!]$ can be represented either by its channel-use vectors $\{\fc_{i,t}\}_{t=1}^{n}$ or by its antenna-wise sequences $\{\fc_{i,\mu}\}_{\mu=1}^{r}$:
	
	\begin{align}
		\fc_i = [c_{i,t,\mu}]
		=
		\begin{bmatrix}
			c_{i,1,1} & c_{i,2,1} & \cdots & c_{i,n,1} \\
			c_{i,1,2} & c_{i,2,2} & \cdots & c_{i,n,2} \\
			\vdots & \vdots & \ddots & \vdots \\
			c_{i,1,r} & c_{i,2,r} & \cdots & c_{i,n,r}
		\end{bmatrix}
		\in \mathbb{R}^{r \times n}.
	\end{align}
	
	\subsection{Decoding}
	To determine if message $j \in [\![M_n(R)]\!]$ was sent, the decoder checks if $\fy$ lies in the decoding set $$\mathbbmss{D}_j = \{ \fy \allowbreak \in \allowbreak \mathbb{R}^{N_T \times n} \,:\; \allowbreak d_F(\fy,\bar{\fc}_j) \leq \delta_n \},$$	where the \emph{decoding measure} between the output $\fy$ and the target codeword $\bar{\fc}_j$ is given by 
	\begin{align}
		d_F(\mathbf{y},\bar{\fc}_j) \allowbreak \triangleq n^{-1} \big\| \mathbf{y}-\mathbf{H} \bar{\fc}_j \big\|_{F}^{2} - N_R \sigma^2,
	\label{Eq.Dec_Measure}
	\end{align}
	and $\delta_n = 4a/3n^{(1 - b)/2}$ is the decoding threshold where $a$ and $b$ are fixed and arbitrary positive constants, respectively.

	\section{Main Results}
	\label{Sec.Res}
	Here, we present our main capacity theorem with the achievability and the converse proofs.
	
	\subsection{Main Results}
	\label{Subsec.MaiN_Results}
	
	\begin{theorem}
		\label{Th.MIMO-Capacity}
		Let \(\bG\) be a MIMO Gaussian channel with $r=\operatorname{rank}(\mathbf H).$ Under a peak-amplitude constraint and a super-exponential codebook-size regime with $M_n(R)=2^{(n\log n)R},$ the identification capacity of \(\bG\) with respect to the adapted first-order rate in \eqref{Eq.First_Order_Rate} is given by
		\begin{align}
			\label{Ineq.LU}
			\frac{r}{4} \leq \mathbb{C}_{\rm I}(\bG) \leq \frac r2 .
		\end{align}
	\end{theorem}
	\begin{proof}
		Proofs for converse and achievability are provided in Subsections~\ref{Subsec.Converse2} and \ref{Subsec.Achievability}, respectively.
	\end{proof}

	Next, in order to present our second theorem, we define the following system parameters:

 	\begin{theorem}
 		\label{Th.MIMO-Capacity2}
 		Let $\bG$ be a MIMO Gaussian channel satisfying the conditions of Theorem \ref{Th.MIMO-Capacity}. Then, for the considered codebook construction and the second-order rate defined in \eqref{Eq.Second_Order_Rate}, $R_2\geq\Lambda_L,$ where
 		\begin{align}
 			\Lambda_L\triangleq
 			r\log\left(
 			\frac{\sigma_r(\mathbf H)\sqrt{P_{\max}}}{\sqrt{2\pi e\,a}}
 			\right).
 		\end{align}
 		Thus, $\Lambda_L$ is achievable order-$n$ contribution relative to the leading coefficient $L_1=(1-b)r/4$ in \eqref{Eq.Second_Order_Rate}.
 	\end{theorem}

	\begin{proof}
		Proof follows by adapting the second order rate provided in \eqref{Eq.Second_Order_Rate} unto \eqref{Eq.Log_M_Achievability}.
	\end{proof}
	
	The achievability and converse bounds derived in Theorem \ref{Th.MIMO-Capacity} have different coefficients for the leading \(n\log n\) term. Specifically, the achievability bound has coefficient $L_1=(1-b)r/4,$ whereas the converse bound has coefficient $r/2.$ Since \(0<b<1\), these coefficients do not coincide. Consequently, a single second-order rate normalization cannot simultaneously eliminate the \(n\log n\) terms from both bounds while retaining finite order-\(n\) remainders. In particular, defining
	\[
	R_2=\liminf_{n\to\infty}
	\frac{\log M_n-L_1n\log n}{n}
	\]
	with \(L_1\) yields a finite achievable lower bound, but the available converse retains a residual term proportional to \(n\log n\). Conversely, choosing \(L_1=r/2\) removes the leading term from the converse but does not yield a finite achievable lower bound under the current achievability result.
	
	Therefore, extracting and comparing the order-\(n\) contributions associated with the system parameters requires tighter bounds with a common leading \(n\log n\) coefficient. This can be achieved either by strengthening the achievability result to match the converse coefficient or, alternatively, by deriving a sharper converse that matches the achievable coefficient. Until such a refinement is established, the current results provide separate bounds on the leading growth rate and the achievable order-\(n\) contribution, but do not establish a finite second-order characterization under a single normalization.

	\begin{customcorollary}{1}[Identity MIMO Gaussian Channel]
		\label{lem:identity-mimo-gaussian}
		Consider the standard real MIMO Gaussian channel $\fY_i = \fX_i + \fZ_i,$
	 	$\forall i \in [\![n]\!]$ where $\fZ_i\sim\mathcal{N}(\mathbf{0},\mathbf{I}_m).$ Under the peak-amplitude constraint and the codebook-size
		$M_n(R)=2^{(n\log n)R}$ of Theorem~\ref{Th.MIMO-Capacity}, the
		identification capacity satisfies
		\begin{align}
			\frac{m}{4}
			\leq \mathbb{C}_{\rm I}(\bG)
			\leq \frac{m}{2}.
		\end{align}
	\end{customcorollary}
	
	\begin{proof}
		For the standard MIMO Gaussian channel, the channel matrix is $\fH = \mathbf{I}_m,$ hence $r = \operatorname{rank}(\fH)=m.$ Substituting $r=m$ into bounds in Theorem~\ref{Th.MIMO-Capacity}, yields
		\begin{align}
			\frac{m}{4}
			\leq \mathbb{C}_{\rm I}(\bG)
			\leq \frac{m}{2} .
		\end{align}
		For special case $m=1$, the standard real SISO Gaussian channel \cite{Salariseddigh_IT,Colomer_2025} is retrieved and satisfies
		\begin{align}
			\frac{1}{4}
			\leq \mathbb{C}_{\rm I}(\bG)
			\leq \frac{1}{2}.
		\end{align}
	\end{proof}

	\subsection{Converse Proof via Fidelity of Output Distributions}
	\label{Subsec.Converse2}
	For brevity in the derivations of Lemma~\ref{Lem.Converse2} and to facilitate the subsequent analysis, we adopt the following conventions:
	\begin{itemize}[leftmargin=*]
		\item $A_{\rm max} \triangleq \sqrt{P_{\rm max}}.$
		\item Define the rank-space amplitude bound $B_{\max}$ as follows
		\begin{align}
			0 < A_{\max} \leq B_{\max}
			\triangleq
			A_{\max}
			\max_{k\in[\![r]\!]}
			\sum_{\mu=1}^{N_T}
			\left|[\fV_r]_{\mu k}\right|
			\leq A_{\max}\sqrt{N_T},
		\end{align}
		where the last inequality follows from the Cauchy--Schwarz inequality
		and the orthonormality of the columns of $\fV_r$.
		Consequently, $|c_{i,t,\nu}| \leq B_{\max}, \, \forall i,t,\nu.$
		\item $\mathbbmss{C}_{\text{\tiny conv}} \triangleq \big\{ \fc_i \in \mathbb{R}^{r \times n}:\; |\bar{c}_{i,t,\nu}| \leq A_{\rm max} , \forall \, i,t,\nu \big\} .$
		\item  $e_0' = 1 - e_0^2$ with $e_0$ defined as $e_0 \triangleq 1 - e_1 - e_2$ and $e_1$ and $e_2$ being fix constant error bounds.
	\end{itemize}
	Let $R$ be an achievable identification rate for $\bG$ and fix arbitrary constants $e_1,e_2>0$ with $e_1+e_2<1$. Let $\{(\mathbbmss{C}_{\text{\tiny conv}}^{(n)},\allowbreak \D^{(n)})\}_{n\in\mathbb{N}}$ be a sequence of $(n,\allowbreak M_n(R),\allowbreak e_1^{(n)},\allowbreak e_2^{(n)})\allowbreak$-codes where $e_1^{(n)} \leq e_1$ and $e_2^{(n)} \leq e_2.$ 
	\begin{customlemma}{1}[Minimum Frobenius Distance]
		\label{Lem.Converse2}
		The codebook $\mathbbmss{C}_{\text{\tiny conv}}$ ensures that any two distinct codewords $\fc_{i_1}$ and $\fc_{i_2}$ in $\mathbbmss{C}_{\text{\tiny conv}}$, with $i_1,i_2 \in [\![M_n(R)]\!]$ and $i_1 \neq i_2$, are separated by a distance of at least
		\begin{align}
			\label{Ineq.Conv_Distance}
			\| \fc_{i_1} - \fc_{i_2} \|_F \geq 2 \alpha \triangleq \frac{2 \sigma \sqrt{\ln (1/e_0')}}{\sigma_{1}(\fH)} .
		\end{align}
	\end{customlemma}
	\begin{proof}
		The proof is provided in Appendix \ref{App.Converse_Proof2}.
	\end{proof}

	We next apply Lemma~\ref{Lem.Converse2} to derive an upper bound on the codebook size. Since the minimum distance of the codebook is $2\alpha$, one can place non-overlapping spheres $\mathcal{S}_{\fc_i}(r,n,\alpha)$ centered at points in $\mathbbmss{C}_{\text{\tiny conv}}$. These spheres are entirely inscribed within the hypercube $\mathbbmss{Q}_{\f0}(r,n,2(B_{\rm max} + \alpha))$. Hence, the number of codewords, $M_n(R),$ is bounded by
	\begin{align}
		\label{Ineq.Codebook_Size_UB}
		M_n(R) = \frac{\text{Vol}_F \left(\bigcup_{i=1}^{M_n(R)} \S_{\fc_i}(r,n,\alpha) \right)}{\text{Vol}_F(\S_{\fc_1}(r,n,\alpha))}
		\stackrel{(a)}{\leq} \frac{\text{Vol}_F \big( \mathbbmss{Q}_{\f0}(r,n,2(B_{\rm max} + \alpha)) \big)}{\text{Vol}_F(\S_{\fc_1}(r,n,\alpha))} ,
	\end{align}
	where $(a)$ uses
	\begin{align}
		\mathbbmss{C}_{\text{\tiny conv}} \subseteq \mathbbmss{Q}_{\f0}(r,n,2(B_{\rm max} + \alpha)) 
		= \big\{ \fc_i \in \mathbb{R}^{r \times n} : |c_{i,t,\nu}| \leq B_{\rm max} + \alpha ,\, \forall i,\ \forall t, \, \forall \nu \big\}.
	\end{align}
	Thereby, \eqref{Ineq.Codebook_Size_UB} can be simplifed as follows
	\begin{align}
		\label{Ineq.Log_M_UB}
		\log M_n(R)
		&
		\leq \log \text{Vol}_F \big( \mathbbmss{Q}_{\f0}(r,n,2(B_{\rm max} + \alpha)) \big) - \text{Vol}_F(\S_{\fc_1}(r,n,\alpha))
		\nonumber\\
		& \stackrel{(a)}{=} \log \text{Vol}_2 \big( \mathbbmss{Q}_{\f0}(r,n,2(B_{\rm max} + \alpha)) \big) - \text{Vol}_2(\S_{\fc_1}(r,n,\alpha))
	\end{align}
	where $(a)$ uses Lemma \ref{Lem.Vol_Preserv}. Next, employing Lemma \eqref{Lem.Neg_Log_Vol} we can establish the following lower bound
	\begin{align}
		& \log M_n(R)
		\nonumber\\&
		\leq
		rn\log\left(2(B_{\max}+\alpha)\right)
		+ \frac{rn}{2}\log\left(\frac{rn}{2}\right)
		- rn\log \alpha
		-\frac{rn}{2}\log(\pi e)
		+o(n)
		\nonumber\\
		&\stackrel{(a)}{=}
		rn\log\left(2(B_{\max}+\alpha)\right) +\frac{rn}{2}\log n
		+\frac{rn}{2}\log\left(\frac{r}{2}\right)
		-rn\log \alpha
		-\frac{rn}{2}\log(\pi e)
		+o(n)
		\nonumber\\[1mm]
		&\stackrel{(b)}{=}
		\frac{r}{2}n\log n
		+rn\log\left(2(B_{\max}+\alpha)\right)
		-rn\log \alpha
		+\frac{rn}{2}
		\log\left(\frac{r}{2\pi e}\right)
		+o(n)
		\nonumber\\[1mm]
		&\stackrel{(c)}{=}
		\frac{r}{2}n\log n +n\bigg[ r\log\left(2(B_{\max}+\alpha)\right)
		-r\log\alpha
		+\frac{r}{2}\log\left(\frac{r}{2\pi e}\right)
		\bigg]
		+o(n) \stackrel{(d)}{\leq} \frac{r}{2}n\log n
		\nonumber\\[1mm]
		& +n\Bigg[
		r\log\left(
		2\left(
		\sqrt{N_TP_{\max}}
		+
		\frac{\sigma\sqrt{\ln(1/e_0')}}{\sigma_1(\fH)}
		\right)
		\right)
		-r\log\left(
		\frac{\sigma\sqrt{\ln(1/e_0')}}{\sigma_1(\fH)}
		\right)
		+\frac{r}{2}\log\left(\frac{r}{2\pi e}\right)
		\Bigg]
		+o(n).
		\label{Ineq.Log_M_UB2}
	\end{align}
	where
	\begin{itemize}
		\item $(a)$ uses
		\[
		\log\left(\frac{rn}{2}\right)
		=
		\log n+\log\left(\frac{r}{2}\right),
		\]
		which gives
		\[
		\frac{rn}{2}\log\left(\frac{rn}{2}\right)
		=
		\frac{rn}{2}\log n
		+
		\frac{rn}{2}\log\left(\frac{r}{2}\right).
		\]
		
		\item $(b)$ combines the two logarithmic terms that are independent
		of $n$ according to
		\[
		\frac{rn}{2}\log\left(\frac{r}{2}\right)
		-\frac{rn}{2}\log(\pi e)
		=
		\frac{rn}{2}\log\left(\frac{r}{2\pi e}\right).
		\]
		
		\item $(c)$ factor outs $n$ from all terms that are linear in $n$,
		yielding
		\begin{align}
		\frac{r}{2}n\log n
		+
		n\bigg[
		r\log\left(2(B_{\max}+\alpha)\right)
		-r\log\alpha
		+\frac{r}{2}\log\left(\frac{r}{2\pi e}\right)
		\bigg]
		+o(n).
		\end{align}
		
		\item $(d)$ substitutes
		\[
		B_{\max}
		\leq
		\sqrt{N_TP_{\max}},
		\qquad
		\alpha
		=
		\frac{\sigma\sqrt{\ln(1/e_0')}}{\sigma_1(\fH)}.
		\]
	\end{itemize}

	\begin{align}
		\label{Eq.Log_M_Converse}
		& \log M_n(R)
		\leq \frac{r}{2}n\log n
		\nonumber\\&
		+n\Bigg[
		r\log\left(
		2\left(
		\sqrt{N_TP_{\max}}
		+
		\frac{\sigma\sqrt{\ln(1/e_0')}}{\sigma_1(\fH)}
		\right)
		\right)
		-r\log\left(
		\frac{\sigma\sqrt{\ln(1/e_0')}}{\sigma_1(\fH)}
		\right)
		+\frac{r}{2}\log\left(\frac{r}{2\pi e}\right)
		\Bigg]
		+o(n) ,
	\end{align}
	where the dominant term for fixed $r$ scales as $n \log n$. Hence, setting $M_n(R) = 2^{(n \log n) R}$, we obtain
	\begin{align}
		\label{Ineq.Rate_UB}
		R \leq \frac{r}{2} + O\left(\frac{1}{\log n}\right).
	\end{align}
	Interpreting $R$ as the corresponding sequence of finite-blocklength coding rates $R_n$, the preceding bound implies
	\begin{equation}
		\limsup_{n\to\infty} R_n \leq \frac{r}{2}.
	\end{equation}
	Consequently, any asymptotically achievable rate $R$ associated with such a sequence satisfies
	\begin{equation}
		R \leq \frac{r}{2}.
	\end{equation}
	It follows that the capacity, defined as the supremum of all asymptotically achievable rates, obey
	\begin{align}
		\mathbb{C}_{\text{I}}(\bG) \leq \frac{r}{2}
	\end{align}
	This completes the converse proof of Theorem~\ref{Th.MIMO-Capacity}. \qed

	\subsection{Achievable Coding Scheme in the Rank-Space}
	\label{Subsec.Achievability}
	
	In the following, we deal with a rank-space codebook $\mathbbmss{C} = \{ \fc_i \} \subset \mathbb{R}^{r \times n},$ with $i \in [\![M_n(R)]\!]$ where each codeword matrix reads $\fc_i = [c_{i,t,\mu}].$ More specifically,
	\begin{align*}
		\mathbbmss{C} \subset \mathbbmss{Q}_{\mathbf 0}(r,n,2A_{\mathrm r})
		\triangleq
		\left\{
		\fc_i\in\mathbb{R}^{r\times n}:
		\|\fc_i\|_{\max}\leq A_{\mathrm r}, \forall i
		\right\}.
	\end{align*}
	where $$A_{\mathrm r} \triangleq \sqrt{\frac{P_{\max}}{r}}$$ denote the half-edge of the rank-space hypercube. We formulate the coding problem in the $r$-dimensional effective signal
	space. Thus, each codeword is generated directly in the $r$-dimensional
	rank space rather than in the full $N_T$-dimensional transmit space. Let $\fc_i\in\mathbb{R}^{r\times n}$ denote a rank-space codeword, and let $\fV_r\in\mathbb{R}^{N_T\times r}$
	denote the fixed linear precoding operator associated with the right-singular
	subspace of $\fH$. Since $\fV_r^{\mathsf T}\fV_r=\mathbf I_r,$
	the mapping
	\begin{equation}
		\bar{\fc}_i \triangleq \fV_r\fc_i \in \mathbb{R}^{N_T\times n}
	\end{equation}
	is an isometric embedding of the rank-space codeword $\fc_i$ into the
	physical $N_T$-dimensional transmit space. In particular,
	$\fV_r$ maps $\fc_i$ into the subspace $$\operatorname{span}(\fV_r)\subseteq\mathbb{R}^{N_T}$$ spanned by the columns of $\fV_r$. Consequently, this mapping preserves Euclidean distances, that is, since $\fV_r^{\mathsf T}\fV_r=\mathbf I_r$, the Frobenius norm
	of the difference between any two physically precoded codewords
	satisfies
	\begin{align}
		\|\bar{\fc}_i-\bar{\fc}_j\|_{\mathrm F}^2
		&=\|\fV_r\fc_i-\fV_r\fc_j\|_{\mathrm F}^2\\
		&=\|\fV_r(\fc_i-\fc_j)\|_{\mathrm F}^2\\
		&=\operatorname{tr}\!\left(
		(\fc_i-\fc_j)^{\mathsf T}
		\fV_r^{\mathsf T}\fV_r
		(\fc_i-\fc_j)
		\right)\\
		&=\operatorname{tr}\!\left(
		(\fc_i-\fc_j)^{\mathsf T}
		(\fc_i-\fc_j)
		\right)\\
		&=\|\fc_i-\fc_j\|_{\mathrm F}^2.
	\end{align}
	Taking square roots yields
	\begin{equation}
			\|\bar{\fc}_i-\bar{\fc}_j\|_{\mathrm F}
			=
			\|\fc_i-\fc_j\|_{\mathrm F},
			\qquad \forall\,i,j.
	\end{equation}
	Accordingly, the codebook is designed in the rank space, whereas the peak-amplitude constraint is enforced on the corresponding antenna-domain codewords after precoding.

	For every $\fc_i\in\mathbbmss{C}$, the peak amplitude of the
	physically precoded codeword $\bar{\fc}_i=\fV_r\fc_i$ satisfies
	\begin{align}
		\|\bar{\fc}_i\|_{\max}
		=\|\fV_r\fc_i\|_{\max}
		& =\max_{1\leq k\leq N_T}
		\|\mathbf v_k^{\mathsf T}\fc_i\|_{\max}
		\nonumber\\
		&\stackrel{(a)}{\leq}
		\max_{1\leq k\leq N_T}
		\|\mathbf v_k\|_2\,\|\fc_i\|_{\max,2}
		\nonumber\\
		&\stackrel{(b)}{\leq}
		\|\fc_i\|_{\max,2}
		\leq \sqrt{r}\,\|\fc_i\|_{\max}
		\nonumber\\
		&\stackrel{(c)}{\leq}
		\sqrt{r}\,A_{\mathrm r}
		=\sqrt{P_{\max}},
	\end{align}
	where $\mathbf v_k^{\mathsf T}$ denotes the $k$-th row of
	$\fV_r$, and
	\[
	\|\fc_i\|_{\max,2}
	\triangleq \max_{1\leq \ell\leq n}
	\|[\fc_i]_{(:,\ell)}\|_2.
	\]
	\begin{itemize}
		\item $(a)$ The Cauchy--Schwarz inequality states that
		$|\mathbf x^{\mathsf T}\mathbf y|
		\leq \|\mathbf x\|_2\|\mathbf y\|_2$ for any real-valued
		vectors $\mathbf x$ and $\mathbf y$. Applying it to each
		column of $\fc_i$ and each row $\mathbf v_k^{\mathsf T}$
		of $\fV_r$ yields the desired inequality.
		
		\item $(b)$ Since $\fV_r^{\mathsf T}\fV_r=\mathbf I_r$,
		the columns of $\fV_r$ are orthonormal, which implies
		$\|\mathbf v_k\|_2\leq 1$ for every row $\mathbf v_k^{\mathsf T}$.
		Moreover, $\|\fc_i\|_{\max,2}\leq
		\sqrt{r}\,\|\fc_i\|_{\max}$ follows because each column
		of $\fc_i$ contains $r$ entries.
		
		\item $(c)$ The rank-space peak-amplitude constraint gives
		$\|\fc_i\|_{\max}\leq A_{\mathrm r}$, where
		$A_{\mathrm r}=\sqrt{P_{\max}/r}$.
	\end{itemize}
	Therefore, the physically precoded codeword satisfies the
	required peak-amplitude constraint.

	\subsection{Rank-Space Signaling}
	
	Recall that $\fV_r\in\mathbb{R}^{N_T\times r}$
	contains the right singular vectors associated with the nonzero singular
	values of $\fH$. The received signal depends only on the projection of the
	transmit codeword onto $\operatorname{span}(\fV_r)$. Indeed, defining $\fP_r\triangleq\fV_r\fV_r^{\mathsf T},$
	we have $\fH=\fH\fP_r,$ and hence, for any transmit codeword $\fX\in\mathbb{R}^{N_T\times n}$, we have
	\begin{equation}
	\fH\fX
	=
	\fH\fP_r\fX
	=
	\fH\fV_r\fV_r^{\mathsf T}\fX.
	\end{equation}
	Thus, the component
	\begin{equation}
		(\mathbf I_{N_T}-\fP_r)\fX
	\in\operatorname{null}(\fH)
	\end{equation}
	is completely invisible at the receiver. Consequently, signaling over
	more than $r$ spatial dimensions cannot provide additional
	information-bearing degrees of freedom, since any component beyond the
	rank-$r$ signal space lies in the channel null space. Accordingly, we formulate the coding problem directly in the
	$r$-dimensional rank space and represent each codeword as
	\[
	\bar{\fc}_i
	\triangleq
	\fV_r \fc_i
	\in\mathbb{R}^{N_T\times n}.
	\]
	Signaling over fewer than $r$ dimensions is possible but does not exploit
	the full spatial signal space supported by the channel. Using all $r$
	dimensions therefore captures the full spatial-multiplexing gain
	associated with the channel rank.

	\begin{customlemma}{2}[Channel-Domain Distance]
		\label{Lem.Min_dis_con_cb}
		
		Let $\fH=\fU_r\fSigma_r\fV_r^{\mathsf T}$ be the reduced SVD of a rank-$r$ channel matrix, and let
		$\fc_i,\fc_j\in\mathbb{R}^{r\times n}$ with physical codewords $\bar{\fc}_k \triangleq \fV_r \fc_k$, $k\in\{i,j\}$.
		Then, the channel-domain distance is lower-bounded by the effective-space distance scaled by $\sigma_r(\fH)$:
		\begin{equation}
			\|\fH(\bar{\fc}_i-\bar{\fc}_j)\|_F
			\geq
			\sigma_r(\fH)\|\fc_i-\fc_j\|_F.
		\end{equation}
	\end{customlemma}
	
	\begin{proof}
	 The proof is provided in Appendix \ref{App.Min_dis_con_cb}.
	\end{proof}

	\subsection{Rate Analysis}
	We use a packing arrangement of non-overlapping \(rn\)-dimensional hyper spheres in the Frobenius norm of radius $\tau_n = \sqrt{n\epsilon_{n}}$ in a \(rn\)-dimensional hyper cube with edge length $P_{\rm max},$ where
	\begin{equation}
		\epsilon_{n} = \frac{a}{\sigma_r^2(\fH) n^{( 1 - b) / 2}},
	\label{Eq.Epsilon_n}
	\end{equation}
	with $a>0$ being a fixed constant and $b\in(0,1)$ being an arbitrary but fixed constant, chosen sufficiently small. Throughout the analysis, $b$ is held fixed independently of the blocklength $n$.

	Let $\mathscr{S}$ denote a sphere packing, i.e., an arrangement of $M$ non-overlapping spheres $\S_{\fc_i}(r,n,\tau_n),\, i \in [\![M_n(R)]\!],$ that are packed inside the larger cube $\mathbbmss{Q}_{\f0}(r,n,2A_{\mathrm r}),$ see Figure \ref{Fig.Density}. We require only that the centers of the spheres lie within the hypercube $\mathbbmss{Q}_{\f0}(r,n,2A_{\mathrm r}),$ that the spheres are mutually disjoint, and that each sphere exhibits a non-empty intersection with $\mathbbmss{Q}_{\f0}(r,n,2A_{\mathrm r}).$ The packing density \cite{CHSN13} is
	\begin{align}
		\Updelta_{r \times n}(\mathscr{S}) \triangleq \frac{\text{Vol}_F \Big(\mathbbmss{Q}_{\f0}(r,n,2A_{\mathrm r}) \cap \bigcup_{i=1}^{M_n(R)}\S_{\fc_i}(r,n,\tau_n) \Big)}{\text{Vol}_F \big(\mathbbmss{Q}_{\f0}(r,n,A_{\mathrm r}) \big)},
		\label{Eq.Def_Density}
	\end{align}
	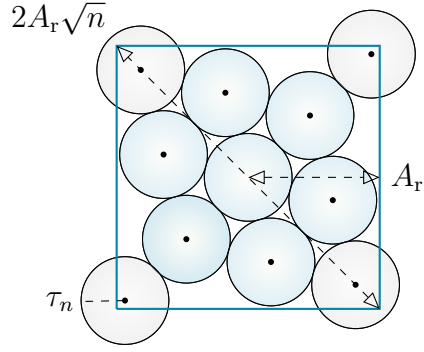
\begin{figure}[t]
		\centering
		\scalebox{1}{
	\begin{tikzpicture}[scale=.58,rotate=0][thick]
		
		\draw (-1.41,-1.41) circle (1cm);
		\draw [inner color = mycolor8!4, outer color =  mycolor7!60, fill opacity=0.7] (-1.41,-1.41) circle (1cm);
		\node [fill=black, shape=circle, inner sep=.8pt] ($.$) at (-1.41,-1.41) {};
		\draw (0,0) circle (1cm);
		\draw [inner color = mycolor8!4, outer color =  mycolor7!60, fill opacity=0.5] (0,0) circle (1cm);
		\draw (1.41,1.41) circle (1cm);
		\draw [inner color = mycolor8!4, outer color =  mycolor7!60, fill opacity=0.5] (1.41,1.41) circle (1cm);
		\node [fill=black, shape=circle, inner sep=.8pt] ($.$) at (1.41,1.41) {};
		\draw (+.52,-1.93) circle (1cm);
		\draw [inner color = mycolor8!4, outer color =  mycolor7!60, fill opacity=0.5] (+.52,-1.93) circle (1cm);
		\node [fill=black, shape=circle, inner sep=.8pt] ($.$) at (+.52,-1.93) {};
		\draw (1.93,-.52) circle (1cm);
		\draw [inner color = mycolor8!4, outer color =  mycolor7!60, fill opacity=0.5] (1.93,-.52) circle (1cm);
		\node [fill=black, shape=circle, inner sep=.8pt] ($.$) at (1.93,-.52) {};
		\draw (-1.93,+.52) circle (1cm);
		\draw [inner color = mycolor8!4, outer color =  mycolor7!60, fill opacity=0.5] (-1.93,.52) circle (1cm);
		\node [fill=black, shape=circle, inner sep=.8pt] ($.$) at (-1.93,.52) {};
		\draw (-.52,1.93) circle (1cm);
		\draw [inner color = mycolor8!4, outer color =  mycolor7!60, fill opacity=0.5] (-.52,1.93) circle (1cm);
		\node [fill=black, shape=circle, inner sep=.8pt] ($.$) at (-.52,1.93) {};

		\draw (2.45,-2.45) circle (1cm);
		\draw [inner color=gray!4,outer color=gray!20, fill opacity=0.5] (2.45,-2.45) circle (1cm);
		\node [fill=black, shape=circle, inner sep=.8pt] ($.$) at (2.45,-2.45) {};
		\draw (2.81,2.81) circle (1cm);
		\draw [inner color=gray!4,outer color=gray!20, fill opacity=0.5] (2.81,2.81) circle (1cm);
		\node [fill=black, shape=circle, inner sep=.8pt] ($.$) at (2.81,2.81) {};
		\draw (-2.45,2.45) circle (1cm);
		\draw [inner color=gray!4,outer color=gray!20, fill opacity=0.5] (-2.45,2.45) circle (1cm);
		\node [fill=black, shape=circle, inner sep=.8pt] ($.$) at (-2.45,2.45) {};
		\draw (-2.81,-2.81) circle (1cm);
		\draw [inner color=gray!4,outer color=gray!20, fill opacity=0.5] (-2.81,-2.81) circle (1cm);
		\node [fill=black, shape=circle, inner sep=.8pt] ($.$) at (-2.81,-2.81) {};
		
		\foreach \s in {3}
		{
			\draw [thick, mycolor12] (-\s,-\s) -- (\s,-\s) -- (\s,\s) -- (-\s,\s) -- (-\s,-\s);
		}
		
		
		\draw [>=open triangle 45,<->,dashed] (0,0) -- (3,0) node [right,font=\normalsize] {$A_{\mathrm r}$};
		\draw [dashed] (-2.81,-2.81) -- (-3.71,-2.83) node [left,font=\normalsize] {$\tau_n$};
		\draw [>=open triangle 45,<->,dashed] (3,-3) -- (-3,3) node [above left,font=\normalsize] {$2A_{\mathrm r}\sqrt{n}$};
\end{tikzpicture}}
		\caption{Illustration of a saturated sphere packing inside a hypercube, where spheres of radius $2\tau_n = \sqrt{n\epsilon_n}$ cover a larger cube. Dark gray colored spheres are not entirely contained within the larger cube, and yet they contribute to the packing arrangement. As we assign a codeword to each sphere center, the $1$-norm and arithmetic mean of a codeword are bounded by $A_{\mathrm r}$ as required.}
		\vspace{-2mm}
		\label{Fig.Density}
	\end{figure}
	We invoke a saturated packing argument as accomplished in \cite{CHSN13}. Specifically, consider a saturated packing of spheres $\bigcup_{i=1}^{M_n(R)} \S_{\fc_i}(r,n,\tau_n)$ with radius $\tau_n = \sqrt{n\epsilon_{n}}$, embedded within the hypercube $\mathbbmss{Q}_{\f0}(r,n,2A_{\mathrm r})$. Rather than pursuing a direct volume-density analysis of the saturated packing, we take a different route. In finite-dimensional bounded geometries, obtaining explicit and sufficiently sharp lower bounds on the packing density, such as the classical \(2^{-n}\)-type bounds \cite{CHSN13}, can be difficult, particularly when the ambient region is a hypercube and the spheres are permitted to intersect its boundary. Such density bounds would require a sufficiently sharp characterization of the geometry of the packing near the boundary and, in general, do not yield a convenient expression for the codebook cardinality. We therefore bypass this geometric difficulty altogether. The key observation is that saturation imposes a stronger and more tractable property on the \emph{centers} of the packing spheres: they necessarily form a covering of the ambient hypercube. This converts the packing problem into a covering problem, for which a simple volume argument directly yields a lower bound on the number of centers, and hence on the number of codewords. We next exploit the saturation property to obtain a covering property of the packing centers.
	
	By construction, all packing centers satisfy
	\begin{equation}
		\fc_i \in \mathbbmss{Q}_{\f0}(r,n,2A_{\mathrm r}),
		\qquad
		i \in [\![M_n(R)]\!],
		\label{Eq:center_in_cube}
	\end{equation}
	and the non-overlap condition implies
	\begin{equation}
		\| \fc_i-\fc_j \|_F
		\geq 2\tau_n,
		\qquad
		i\neq j.
		\label{Eq:packing_separation}
	\end{equation}
	The saturation of the packing means that no additional sphere of radius $\tau_n$, whose center lies in $\mathbbmss{Q}_{\f0}(r,n,2A_{\mathrm r})$, can be added while preserving the non-overlap condition. This property implies that the collection of packing centers constitutes a $2\tau_n$-covering of the hypercube. To see this, suppose, to the contrary, that there exists a point $\fx \in \mathbbmss{Q}_{\f0}(r,n,A_{\mathrm r})$ such that
	\begin{equation}
		\| \fx-\fc_i \|_F
		> 2\tau_n,
		\qquad
		\forall i\in[\![M_n(R)]\!].
		\label{Eq:uncovered_point}
	\end{equation}
	Since $\fx$ belongs to $\mathbbmss{Q}_{\f0}(r,n,A_{\mathrm r})$, a sphere $\S_{\fx}(r,n,\tau_n)$ centered at $\fx$ is an admissible additional sphere. Moreover, \eqref{Eq:uncovered_point} implies that the distance between the centers of $\S_{\fx}(r,n,\tau_n)$ and $\S_{\fc_i}(r,n,\tau_n)$ is strictly larger than $2\tau_n$ for every $i\in[\![M_n(R)]\!]$. Hence, $\forall i\in[\![M_n(R)]\!]$ it holds
	\begin{equation}
		\S_{\fx}(r,n,\tau_n)
		\cap
		\S_{\fc_i}(r,n,\tau_n)
		=
		\varnothing .
		\label{Eq:new_sphere_disjoint}
	\end{equation}
	Thus, $\S_{\fx}(r,n,\tau_n)$ can be added to the packing without violating the non-overlap condition, which contradicts the saturation assumption. Therefore,
	\begin{equation}
		\forall \fx\in\mathbbmss{Q}_{\f0}(r,n,2A_{\mathrm r}),
		\qquad
		\min_{i\in[\![M_n(R)]\!]}
		\| \fx-\fc_i \|_F
		\leq 2\tau_n.
		\label{Eq:2tau_covering}
	\end{equation}
	Equivalently,
	\begin{equation}
		\mathbbmss{Q}_{\f0}(r,n,2A_{\mathrm r})
		\subseteq
		\bigcup_{i=1}^{M_n(R)}
		\S_{\fc_i}(r,n,2\tau_n).
		\label{Eq:covering}
	\end{equation}
	Thus, every saturated packing of radius-$\tau_n$ spheres induces a $2\tau_n$-covering of the entire hypercube. We can now lower-bound the number of packing centers by comparing volumes. Taking $\text{Vol}_F$ on both sides of \eqref{Eq:covering} gives
	\begin{equation}
		\text{Vol}_F
		\left(
		\mathbbmss{Q}_{\f0}(r,n,2A_{\mathrm r})
		\right)
		\leq
		\text{Vol}_F
		\left(
		\bigcup_{i=1}^{M_n(R)}
		\S_{\fc_i}(r,n,2\tau_n)
		\right).
		\label{Eq:volume_covering}
	\end{equation}
	Applying subadditivity of volume to the union on the right-hand side yields
	\begin{equation}
		\text{Vol}_F
		\left(
		\mathbbmss{Q}_{\f0}(r,n,2A_{\mathrm r})
		\right)
		\leq
		\sum_{i=1}^{M_n(R)}
		\text{Vol}_F
		\left(
		\S_{\fc_i}(r,n,2\tau_n)
		\right).
		\label{Eq:volume_subadditivity}
	\end{equation}
	Since all covering spheres have the same radius, their volumes are identical and independent of their centers. Hence,
	\begin{equation}
		\text{Vol}_F
		\left(
		\mathbbmss{Q}_{\f0}(r,n,2A_{\mathrm r})
		\right)
		\leq
		M_n(R) \cdot 
		\text{Vol}_F
		\left(
		\S_{\f0}(r,n,2\tau_n)
		\right).
		\label{Eq:volume_comparison}
	\end{equation}
	Consequently,
	\begin{equation}
		M_n(R)
		\geq
		\frac{
			\text{Vol}_F
			\left(
			\mathbbmss{Q}_{\f0}(r,n,2A_{\mathrm r})
			\right)
		}{
			\text{Vol}_F
			\left(
			\S_{\f0}(r,n,2\tau_n)
			\right)
		}.
		\label{Eq:M_lower_bound_volume}
	\end{equation}

	Next, we proceed to establish the connection between the Frobenius geometry on
	$\mathbb{R}^{r\times n}$ and the standard Euclidean geometry on
	$\mathbb{R}^{rn}$. In particular, the vectorization map is an isometry between
	these two spaces, and therefore preserves distances, norms, and the
	corresponding Lebesgue volumes. Consequently, Frobenius-norm balls and
	hyperspheres in $\mathbb{R}^{r\times n}$ can be identified, via vectorization,
	with the usual Euclidean balls and hyperspheres in $\mathbb{R}^{rn}$. To establish the lemma, we first recall the following definitions and
	preliminaries concerning the Frobenius and Euclidean structures, the
	corresponding Lebesgue measures, and the vectorization map.

	\begin{itemize}[leftmargin=*]
		\item \textbf{Frobenius inner product.}
		For $\fX,\fY\in\mathbb{R}^{r\times n}$, the Frobenius inner product is
		defined by
		\[
		\langle \fX,\fY\rangle_F
		\triangleq
		\operatorname{tr}(\fX^{\mathsf T}\fY)
		=
		\sum_{a=1}^{r}\sum_{b=1}^{n}x_{ab}y_{ab}.
		\]
		
		\item \textbf{Frobenius norm.}
		The norm induced by the Frobenius inner product is
		\[
		\|\fX\|_F
		\triangleq
		\sqrt{\langle \fX,\fX\rangle_F}
		=
		\left(
		\sum_{a=1}^{r}\sum_{b=1}^{n}x_{ab}^{2}
		\right)^{1/2}.
		\]
		
		\item \textbf{Euclidean inner product and norm.}
		For $\fx,\fy \in\mathbb{R}^{rn}$, the standard Euclidean inner
		product and norm read  
		\[
		\langle \fx,\fy \rangle_2
		\triangleq
		\fx^{\mathsf T}\fy
		=
		\sum_{k=1}^{rn}x_k y_k,
		\]
		and
		\[
		\| \fx \|_2 \triangleq \sqrt{\langle \fx,\fx \rangle_2}	= \left(
		\sum_{k=1}^{rn}x_k^2
		\right)^{1/2}.
		\]
		\item \textbf{Frobenius Lebesgue volume.}
		Let $\operatorname{Vol}_F$ denote the $rn$-dimensional Lebesgue
		measure on $\mathbb{R}^{r\times n}$ induced by the Frobenius inner
		product $\langle\cdot,\cdot\rangle_F$.
		\item \textbf{Euclidean Lebesgue volume.}
		Let $\operatorname{Vol}_2$ denote the standard $rn$-dimensional
		Lebesgue measure on $\mathbb{R}^{rn}$ induced by the standard
		Euclidean inner product $\langle\cdot,\cdot\rangle_2$.
		\item \textbf{Vectorization map.}
		Let
		\begin{equation}
			\operatorname{vec}:
			\bigl(\mathbb{R}^{r\times n},\|\cdot\|_F\bigr)
			\longrightarrow
			\bigl(\mathbb{R}^{rn},\|\cdot\|_2\bigr)
		\end{equation}
		be the column-wise vectorization map. For
		$\fX=(x_{ij})\in\mathbb{R}^{r\times n}$, it is defined by
		\begin{equation}
			\operatorname{vec}(\fX)
			=
			\begin{pmatrix}
				\fX_{:,1}\\
				\fX_{:,2}\\
				\vdots\\
				\fX_{:,n}
			\end{pmatrix}
		\end{equation}
		or, equivalently, componentwise by
		\begin{equation}
			\bigl[\operatorname{vec}(\fX)\bigr]_{(j-1)r+i}
			=
			x_{ij},
			\qquad
			1\leq i\leq r,\quad
			1\leq j\leq n.
		\end{equation}
	\end{itemize}

	\begin{customlemma}{3}[Volume Preservation under Vectorization]
		\label{Lem.Vol_Preserv}
		For every Lebesgue-measurable set
		$\mathcal{A}\subseteq\mathbb{R}^{r\times n}$,
		the Frobenius volume of $\mathcal{A}$ coincides with the Euclidean
		volume of its vectorized image:	$\operatorname{Vol}_{F}(\mathcal{A}) = \operatorname{Vol}_{2}\!\left(\operatorname{vec}(\mathcal{A})\right)$ where
		\begin{equation}
			\operatorname{vec}:
			\bigl(\mathbb{R}^{r\times n},\|\cdot\|_F\bigr)
			\longrightarrow
			\bigl(\mathbb{R}^{rn},\|\cdot\|_2\bigr) .
		\end{equation}
	\end{customlemma}
	
	\begin{proof}
		The proof is provided in Appendix \ref{App.Vol_Preserv}.
	\end{proof}

	In particular, for
	\begin{equation}
		\S_{\fx_0}(r,n,\tau)
		=
		\left\{
		\fx\in\mathbb{R}^{r\times n}:
		\|\fx-\fx_0\|_F\leq\tau
		\right\},
	\end{equation}
	the distance-preserving property of the employed vectorization in Lemma \ref{Lem.Vol_Preserv} gives
	\begin{equation}
		\operatorname{vec}\!\left(\S_{\fx_0}(r,n,\tau)\right)
		=
		\S_{\operatorname{vec}(\fx_0)}(rn,\tau),
	\end{equation}
	where
	\begin{equation}
		\S_{\fx_0}(rn,\tau)
		=
		\left\{
		\fx\in\mathbb{R}^{rn}:
		\|\fx-\fx_0\|_2\leq\tau
		\right\}.
	\end{equation}
	Consequently,
	\begin{equation}
		\label{Eq.Vol_Equi}
		\operatorname{Vol}_{F}\!\left(\S_{\fx_0}(r,n,\tau)\right)
		=
		\operatorname{Vol}_{2}\!\left(
		\S_{\operatorname{vec}(\fx_0)}(rn,\tau)
		\right).
	\end{equation}
	
	We observe that, the volume of an $n$-dimensional hypersphere of radius \( \tau \) is given by \cite[Eq.~(16)]{CHSN13},
	\begin{align}
		\operatorname{Vol}_{2}\big(\mathcal S_{\mathbf c_i}(n,r)\big) = \frac{(\sqrt{\pi}\,r)^n}{\Gamma\!\left( 1 + \frac{n}{2} \right)}
		\label{Eq.VolS} .
	\end{align}
	Hence, combining \eqref{Eq.Vol_Equi} and \eqref{Eq.VolS} yields
	\begin{align}
		\label{Eq.Vol_F}
		\operatorname{Vol}_{F}\!\left(\S_{\fc_i}(r,n,\tau)\right) & = \operatorname{Vol}_{2}\big(\mathcal S_{\operatorname{vec}(\fc_i)}(rn,\tau)\big)
		\nonumber\\&
		= \frac{(\sqrt{\pi}\,\tau)^{rn}}{\Gamma\!\left(1 + \frac{rn}{2}\right)} .
	\end{align}
	Furthermore, the volume of an \(rn\)-dimensional hypercube with side length
	\(A_{\mathrm r}\) is
	\begin{equation}
		\operatorname{Vol}_{2}
		\big(
		\mathbbmss{Q}_{\f0}(n,2A_{\mathrm r})
		\big)
		=
		(2A_{\mathrm r})^{n}.
	\end{equation}
	By Lemma \ref{Lem.Vol_Preserv},
	\begin{align}
		\label{Eq.Vol_Q_F}
		\operatorname{Vol}_{F}
		\left(
		\mathbbmss{Q}_{\f0}(r,n,2A_{\mathrm r})
		\right)
		&=
		(2A_{\mathrm r})^{rn}.
	\end{align}

	Next, we observe given that each hyper sphere has volume $\mathrm{Vol}(\mathcal{S}_{\mathbf{c}_1}(r,n,\tau_n))$ and all centers lie within $\mathbbmss{Q}_{\mathbf{0}}(r,n,2A_{\mathrm r})$, recalling \eqref{Eq:M_lower_bound_volume}, the number of packed spheres, $M$, reads
	\begin{equation}
		M_n(R)
		\geq
		\frac{
			\text{Vol}_F
			\left(
			\mathbbmss{Q}_{\f0}(r,n,2A_{\mathrm r})
			\right)
		}{
			\text{Vol}_F
			\left(
			\S_{\f0}(r,n,2\tau_n)
			\right)
		}.
		\label{Eq:M}
	\end{equation}
	The bound in \eqref{Eq:M} admits the following simplification
	\begin{align}
		\log M_n(R)
		& \geq
		\log
		\text{Vol}_F
		\left(
		\mathbbmss{Q}_{\f0}(r,n,2A_{\mathrm r})
		\right)
		-
		\log
		\text{Vol}_F
		\left(
		\S_{\f0}(r,n,2\tau_n)
		\right).
		\label{Eq.Log_M_0}
	\end{align}
	To further simplify \eqref{Eq.Log_M_0}, we present the following lemma establishing a lower bound on the negative log-volume of an $n$-dimensional hypersphere.

	\begin{customlemma}{4}[Log-Volume of Hypersphere]
		\label{Lem.Neg_Log_Vol}
		Let $n \geq 1$ and $r_0 >0.$ Then, the negative log-volume of $\S_{\fc_1}(n,r_0)$ reads
		\begin{align}
			- \log \text{Vol}(\S_{\fc_1}(n,r_0)) 
			= \frac n2\log\frac n2
			-n\log r_0
			-\frac n2\log(\pi e)
			+o(n).
		\end{align}
	\end{customlemma}
	\begin{proof}
		The proof is provided in Appendix \ref{App.Neg_Log_Vol}.
	\end{proof}

	Combining \eqref{Eq.Log_M_0} with the lower bound on the negative logarithm of the hypersphere volume established in Lemma \ref{Lem.Neg_Log_Vol}, we obtain the following lower bound on the logarithm of the achievable codebook cardinality:
	\begin{align}
		\label{Eq.Log_M}
		\log M_n(R)
		& \geq
		\log
		\text{Vol}_F
		\left(
		\mathbbmss{Q}_{\f0}(r,n,2A_{\mathrm r})
		\right)
		-
		\log
		\text{Vol}_F
		\left(
		\S_{\f0}(r,n,2\tau_n)
		\right)
		\nonumber\\
		& \stackrel{(a)}{=} \log
		\text{Vol}_2
		\left(
		\mathbbmss{Q}_{\f0}(rn,2A_{\mathrm r})
		\right)
		-
		\log
		\text{Vol}_2
		\left(
		\S_{\f0}(rn,2\tau_n)
		\right)
		\nonumber\\
		& \stackrel{(b)}{=}
		\log \text{Vol}_2
		\left(
		\mathbbmss{Q}_{\f0}(rn,2A_{\mathrm r})
		\right)
		-
		\log
		\text{Vol}_2
		\left(
		\S_{\f0}(rn,\tau_n)
		\right) - rn
		\nonumber\\&
 		\stackrel{(c)}{=} rn \log A_{\mathrm r}  + \frac{rn}{2}\log\frac{rn}{2}
 		-rn\log \tau_n
 		-\frac{rn}{2}\log(\pi e)
 		+o(n)
	\end{align}
	where
	\begin{itemize}
		\item $(a)$ exploits Lemma \ref{Lem.Vol_Preserv} and \eqref{Eq.Vol_Equi}. 
		\item $(b)$ employs \eqref{Eq.Vol_F}. 
		\item $(c)$ uses Lemma \eqref{Lem.Neg_Log_Vol} with setting $n = rn.$
	\end{itemize}Next, we observe that employing \eqref{Eq.Epsilon_n}, we obtain
	\begin{align}
		\tau_n = \sqrt{n\epsilon_{n}} = \frac{\sqrt{a}n^{(1 + b) / 4}}{\sigma_r(\fH)} .
	\end{align}
	Thereby, \eqref{Eq.Log_M} is simplified to the following
	\begin{align}
		\log M_n(R)
		& \geq rn \log A_{\mathrm r}  + \frac{rn}{2}\log\frac{rn}{2}
		-rn\log \tau_n
		-\frac{rn}{2}\log(\pi e)
		+o(n)
		\nonumber\\
		& = rn \log A_{\mathrm r} + \frac{rn}{2}\log\frac{rn}{2} - rn \log \frac{\sqrt{a}n^{(1 + b) / 4}}{\sigma_r(\fH)}
		-\frac{rn}{2}\log(\pi e) + o(n) .
	\end{align} 
	Therefore, the logarithm of the achievable codebook cardinality admits the following bound
	\begin{align}
		\label{Eq.Log_M_Achievability}
		\log M_n(R) & \geq \frac{1-b}{4}\,rn\log rn
		+
		rn\log\left(
		\frac{\sigma_r(\fH)\sqrt{P_{\max}}}{\sqrt{2\pi e\,a}\;r^{(1-b)/4}}
		\right)
		+o(n).
	\end{align}
	Since $r$ is fixed, the dominant term in the above lower bound scales as $n\log n$. Hence, setting $M_n(R)=2^{(n\log n)R},$
	we obtain
	\begin{align}
		\label{Ineq.Rate_LB}
		R
		&\geq
		\frac{r(1-b)}{4}
		+
		O\left(\frac{1}{\log n}\right).
	\end{align}
	Interpreting $R$ as the corresponding sequence of finite-blocklength
	coding rates $R_n$, the preceding bound implies
	\begin{equation}
		\liminf_{n\to\infty} R_n
		\geq
		\frac{r(1-b)}{4}.
	\end{equation}
	Consequently, the corresponding asymptotically achievable rate $R$
	satisfies
	\begin{equation}
		R
		\geq
		\frac{r(1-b)}{4}.
	\end{equation}
	Since $b>0$ can be chosen arbitrarily small, letting $b\downarrow 0$
	yields
	\begin{equation}
		R\geq \frac{r}{4}.
	\end{equation}
	Therefore, defining the identification capacity as the supremum of
	all asymptotically achievable rates, we obtain
	\begin{align}
		\mathbb{C}_{\mathrm{I}}(\bG)
		\geq
		\frac{r}{4}.
	\end{align}
	This establishes the lower bound on the identification capacity and completes the achievability proof of Theorem \ref{Th.MIMO-Capacity}.

	\textbf{\textcolor{mycolor12}{Error Analysis.}}
	Let $e_1, e_2, \eta_0, \zeta_0, \zeta_1 > 0$ be arbitrarily small constants. Before proceeding, we set the following conventions to ensure a clear and focused analysis:
	\begin{itemize}[leftmargin=*]
		\item $\fY(i) = \fH \bar{\fc}_i + \fZ \in \mathbb{R}^{N_R \times n}$ denotes the received signal matrix conditioned on $\fx = \bar{\fc}_i.$
		\item $\bar{\fd}_{i,j}^{\fH} \triangleq \fH ( \bar{\fc}_i - \bar{\fc}_j )$ is the transformed codeword difference.
		\item $U_{i,j} \triangleq n^{-1} \big( \big\| \fZ \big\|_F^2 + \big\| \bar{\fd}_{i,j}^{\fH} \big\|_{F}^2 \big),\ V_{i,j} \triangleq 2n^{-1}\big\langle \fZ,\, \bar{\fd}_{i,j}^{\fH} \big\rangle_F.$
		\item $W_{i,j} \triangleq U_{i,j} + V_{i,j}.$
		\item $\mathbbmss{E}_0 \triangleq \{  \fZ \in \mathbb{R}^{N_R \times n} : |V_{i,j}| > \delta_n \}.$
		\item $\mathbbmss{E}_1 \triangleq \{  \fZ \in \mathbb{R}^{N_R \times n} : U_{i,j} - N_R \sigma^2 \leq 2\delta_n \}.$
		\item $\mathbbmss{E}_2 \triangleq \{  \fZ \in \mathbb{R}^{N_R \times n} : W_{i,j} - N_R \sigma^2 \leq \delta_n \}.$
	\end{itemize}

	\textbf{\textcolor{mycolor12}{Type I Error.}}
	The type I errors occur when the transmitter sends $\fc_i,$ yet $\fY \notin \mathbbmss{D}_i.$ For every $i \in [\![M_n(R)]\!],$ the type I error probability is bounded by 
	\begin{align}
		\label{Eq.TypeIError}
		P_{e,1}(i) \leq \Pr\big( \fY(i) \in \mathbbmss{D}_i^c \big)
		= \Pr\big( d_F( \fY(i),\bar{\fc}_i ) > \delta_n \big) .
	\end{align}
	To bound $P_{e,1}(i),$ we perform the one-sided form of Chebyshev's inequality, namely,
	\begin{align}
		\label{Ineq.TypeI_Cheb}
		P_{e,1}(i) & \leq \Pr\big( \big| d_F(\fY(i),\bar{\fc}_i) - \mathbb{E} \big[ d_F(\fY(i),\bar{\fc}_i) \big] \big| > \delta_n \big)
		\nonumber\\&
		\le \frac{\text{Var}[d_F(\fY(i),\bar{\fc}_i)]}{\delta_n^2} .
	\end{align}
	Now, oberve that the expectation of the decoding measure reads
	\begin{align}
		\label{Ineq.Exp_Decoding_Metric}
		\mathbb{E} \big[ d_F( \fY(i),\bar{\fc}_i ) \big] & \stackrel{(a)}{=} \mathbb{E} \Big[ n^{-1} \big\| \fY(i) - \fH \bar{\fc}_i \big\|_{F}^2 \Big] - N_R\sigma^2
		\nonumber\\&
		\stackrel{(b)}{=} \mathbb{E} \bigg[ n^{-1} \sum_{t=1}^n \sum_{\nu=1}^{N_R}  Z_{i,t,\nu}^2 \bigg] - N_R\sigma^2
		\nonumber\\&
		\stackrel{(c)}{=} n^{-1} \sum_{t=1}^n \sum_{\nu=1}^{N_R} \mathbb{E} [ Z_{i,t,\nu}^2 ] - N_R\sigma^2
		\nonumber\\&
		= \frac 1n \cdot n N_R\sigma^2 - N_R\sigma^2 = 0 ,
	\end{align}
	where
	\begin{itemize}
		\item $(a)$ uses \eqref{Eq.Dec_Measure}, with setting $\fY = \fY(i).$
		\item $(b)$ follows from the Frobenius norm expansion.
		\item $(c)$ uses linearity of the expectation.
		\item $(d)$ exploits $\operatorname{Var}(X) = \mathbb{E}[X^2]-\mathbb{E}[X]^2$ with $X=Z_{i,t,\nu}^2$ and $Z_{i,t,\mu} \overset{\text{i.i.d.}}{\sim} \mathcal{N}(0,\sigma^2),\ \forall i \in [\![M_n(R)]\!].$
	\end{itemize}
	
	Next, we compute the variance of the decoding measure as follows
	\begin{align}
		\label{Ineq.Var_Decoding_Metric}
		\operatorname{Var}\!\left[d_F(\fY(i),\bar{\fc}_i)\right] &
		\stackrel{(a)}{=}
		\operatorname{Var}\!\left[
		n^{-1}\left\|\fY(i)-\fH\bar{\fc}_i\right\|_F^2
		\right]
		\nonumber\\&
		\stackrel{(b)}{=}
		\operatorname{Var}\!\left[
		\frac{1}{n}
		\sum_{t=1}^{n} \sum_{\nu=1}^{N_R}
		Z_{i,t,\nu}^2
		\right]
		\nonumber\\&
		\stackrel{(c)}{=}
		\frac{1}{n^2}
		\sum_{t=1}^{n} \sum_{\nu=1}^{N_R}
		\operatorname{Var}\!\left[Z_{i,t,\nu}^2\right]
		\nonumber\\&
		\stackrel{(d)}{=}
		\frac{1}{n^2}
		\sum_{t=1}^{n} \sum_{\nu=1}^{N_R}
		\left(3\sigma^4-\sigma^4\right)
		\nonumber\\&
		\stackrel{(e)}{=}
		\frac{1}{n^2}(nN_R)(2\sigma^4) = \frac{2N_R\sigma^4}{n} ,
	\end{align}
	where
	\begin{itemize}
		\item $(a)$ uses \eqref{Eq.Dec_Measure}.
		\item $(b)$ follows from the Frobenius norm expansion, i.e., $\|\fZ_i\|_F^2 = \sum_{t=1}^{n} \sum_{\nu=1}^{N_R} Z_{i,t,\nu}^2.$
		\item $(c)$ holds since the i.i.d. noise samples $\{Z_{i,t,\nu}\}_{\nu,t}$ imply independence of $\{Z_{i,t,\nu}^2\}_{\nu,t}.$
		\item $(d)$ exploits $\operatorname{Var}(X) = \mathbb{E}[X^2]-\mathbb{E}[X]^2$ with $X=Z_{i,t,\nu}^2$ and $Z_{i,t,\mu} \overset{\text{i.i.d.}}{\sim} \mathcal{N}(0,\sigma^2),\ \forall i \in [\![M_n(R)]\!].$
		\item $(e)$ holds since there are $nN_R$ terms in the double sum.
	\end{itemize}

	Therefore, combining \eqref{Ineq.Exp_Decoding_Metric} and \eqref{Ineq.Var_Decoding_Metric} gives
	\begin{align}
		\label{Ineq.TypeI_Cheb_2}
		P_{e,1}(i) \leq \Pr\big( \big| d_F(\fY(i),\bar{\fc}_i) - \mathbb{E} \big[ d_F(\fY(i),\bar{\fc}_i) \big] \big| \ge \delta_n \big)
		\le \frac{\text{Var}[d_F(\fY(i),\bar{\fc}_i)]}{\delta_n^2}
		\leq \frac{2N_R\sigma^4}{n \delta_n^2}
		\triangleq \eta_0 .
	\end{align}
	Hence, $P_{e,1}(i) \leq \eta_0 \leq e_1$ holds for sufficiently large $n$ and arbitrarily small $e_1 > 0.$

	\textbf{\textcolor{mycolor12}{Type II Error.}}
	We examine type II errors, i.e., when $\fY \in \mathbbmss{D}_j$ while the transmitter sent $\fc_i$ with $i \neq j \,.$ Then, for every $i,j \in [\![M_n(R)]\!],$ the type II error probability reads
	\begin{align}
		P_{e,2}(i,j) = \Pr \big( d_F( \fY(i),\bar{\fc}_j ) \leq \delta_n ).
		\label{Eq.Pe2G}
	\end{align}
	Next, we decompose the decoding measure in \eqref{Eq.Pe2G} as follows
	\begin{align}
		\label{Eq.d_F_Decompos}
		d_F( \fY(i),\bar{\fc}_j ) & = n^{-1} \big\| \fY(i) - \fH \bar{\fc}_j \big\|_{F}^2 - N_R \sigma^2
		\nonumber\\&
		= n^{-1} \big\| \fY(i) - \fH \bar{\fc}_i + \fH \bar{\fc}_i - \fH \bar{\fc}_j \big\|_{F}^2 - N_R \sigma^2
		\nonumber\\&
		= n^{-1} \big\| \fZ + \bar{\fd}_{i,j}^{\fH} \big\|_{F}^2 - N_R \sigma^2
		\nonumber\\&
		=  n^{-1} \big\| \fZ \big\|_F^2 + n^{-1} \big\| \bar{\fd}_{i,j}^{\fH} \big\|_{F}^2 + 2n^{-1} \big\langle \fZ,\, \bar{\fd}_{i,j}^{\fH} \big\rangle_F - N_R \sigma^2
		\nonumber\\&
		= U_{i,j} + V_{i,j} - N_R \sigma^2
		\nonumber\\&
		= W_{i,j} - N_R \sigma^2 .
	\end{align}
	Therefore,
	\begin{align}
		\label{Eq.typeII_decom}
		P_{e,2}(i,j) = \Pr \big( d_F( \fY(i),\bar{\fc}_j ) \leq \delta_n )
		= \Pr \big( W_{i,j} - N_R \sigma^2 \leq \delta_n \big) = \Pr(\mathbbmss{E}_2) .
	\end{align}
	Next, we proceed to show that the cross-term in \eqref{Eq.d_F_Decompos} has negligible probability. To this purpose, we first provide a useful lemma which characterize its distribution.
	\begin{customlemma}{5}[Noise-Channel-Distance Product]
	\label{Lem.Noise_Ch_Dis_Prod}
		
		The cross term $V_{i,j} \triangleq 2n^{-1} \big\langle \mathbf{Z},\,\mathbf{d}_{i,j}^{H} \big\rangle_F$ is distributed as
		\begin{equation}
			\label{eq:cross_term_distribution}
			V_{i,j} \sim \mathcal{N} \big( 0,\, 4n^{-2} \sigma^2 \| \bar{\fd}_{i,j}^{\fH} \|_F^2 \big).
		\end{equation}
	\end{customlemma}

	\begin{proof}
		The proof is provided in Appendix \ref{App.Noise_Ch_Dis_Prod}.
	\end{proof}

	Next, by the submultiplicativity of the Frobenius norm with respect to the induced spectral norm,
	\begin{align}
		\label{Eq.Fro_Nor_d}
		\hspace{-4mm} \big\| \bar{\fd}_{i,j}^{\fH} \big\|_F
		= \left\|\fH( \bar{\fc}_i - \bar{\fc}_j)\right\|_F
		\stackrel{(a)}{\leq}
		\|\fH\|_2\, \| \bar{\fc}_i - \bar{\fc}_j \|_F
		= \sigma_{1}(\fH)\, \| \bar{\fc}_i - \bar{\fc}_j \|_F
		\stackrel{(b)}{\leq}
		\sigma_{1}(\fH)
		\left(\|\bar{\fc}_i\|_F+\|\bar{\fc}_j\|_F\right) , \hspace{-3mm}
	\end{align}
	where
	\begin{itemize}
		\item $(a)$ uses the spectral-Frobenius norm inequality bound, i.e., $\|\fA\fB\|_F \leq \|\fA\|_2\,\|\fB\|_F$ for compatible matrices $\fA \in \mathbb{R}^{a \times b}$ and $\fB \in \mathbb{R}^{c \times d}.$
		\item $(b)$ follows from the triangle inequality for the Frobenius norm, i.e., $\| \bar{\fc}_i - \bar{\fc}_j \|_F \leq \| \bar{\fc}_i \|_F + \| \bar{\fc}_j \|_F .$
	\end{itemize}
	Furthermore, we observe that the entrywise peak constraint implies that $\forall k \in \{i,j\}.$
	\begin{align}
		\label{Eq.Fro_Nor_c}
		\|\bar{\fc}_k\|_F = \sqrt{ \sum_{t=1}^{n}\sum_{\mu=1}^{N_T} \left|\bar{c}_{k,t,\mu}\right|^2 }
		\leq \sqrt{ \sum_{t=1}^{n}\sum_{\mu=1}^{N_T} P_{\max}^2 }
		= P_{\max} \sqrt{N_Tn} .
	\end{align}
	Consequently, combining \eqref{Eq.Fro_Nor_d} and
	\eqref{Eq.Fro_Nor_c} yields
	\begin{align}
		\label{Eq.Forb_d_UB}
		\big\| \bar{\fd}_{i,j}^{\fH} \big\|_F^2 \leq 4 \sigma_{1}^2(\fH) P_{\max}^2 N_T n .
	\end{align}
	
	To apply a standard Gaussian tail bound, we first standardize $V_{i,j}$ by its standard deviation. In particular, if $$V_{i,j} \sim \mathcal{N} \big( 0,\, 4n^{-2} \sigma^2 \| \bar{\fd}_{i,j}^{\fH} \|_F^2 \big)$$ then $\operatorname{std}(V_{i,j}) = 2 \sigma n^{-1} \|\mathbf d_{i,j}^{H} \|_F.$ Defining
	\begin{equation}
		\label{Eq.Stand_V}
		\bar{V}_{i,j}
		\triangleq
		\frac{V_{i,j}}{\operatorname{std}(V_{i,j})},
	\end{equation}
	we have $\bar{V}_{i,j} \sim \mathcal{N}(0,1).$ This standardization allows us to express the tail probability of $V_{i,j}$ in terms of the standard Gaussian tail function $Q(\cdot)$, for which standard tail bounds are readily available.
	Now, recall the event $\mathbbmss{E}_0:$
	\begin{equation}
		\mathbbmss{E}_0 \triangleq \{ |V_{i,j}| > \delta_n \} = \big\{ \big| 2 n^{-1} \big\langle \fZ,\, \bar{\fd}_{i,j}^{\fH} \big\rangle_F \big| > \delta_n \big\}.
	\end{equation}
	Hence, employing \eqref{Eq.Stand_V} we obtain
	\begin{align}
		\Pr(\mathbbmss{E}_0) & = \Pr( |\operatorname{std}(V_{i,j}) \cdot \bar{V}_{i,j} | > \delta_n ) 
		\nonumber\\&
		= \Pr\Big( | \bar{V}_{i,j} | > \frac{\delta_n}{\operatorname{std}(V_{i,j})} \Big)
		\nonumber\\&
		= \Pr\bigg( | \bar{V}_{i,j} | > \frac{n \delta_n}{2 \sigma \|\mathbf d_{i,j}^{H} \|_F} \bigg) .
	\end{align}
	Next, since $\bar{V}_{i,j}\sim\mathcal{N}(0,1)$ and the standard normal distribution is symmetric, for $x>0$ it holds
	\begin{equation}
		\Pr\left(|\bar{V}_{i,j}|>x\right) = 2Q(x) .
	\end{equation}
	Thereby, using the standard Gaussian tail bound, valid for $x\geq 0$, $Q(x)\leq 0.5 e^{-x^2/2},$ we obtain
	\begin{align}
		\Pr(\mathbbmss{E}_0) \leq 2Q \bigg( \frac{n \delta_n}{2 \sigma \|\mathbf d_{i,j}^{H} \|_F} \bigg)
		\leq \exp \biggl( - \frac{1}{2} \bigg[ \frac{n \delta_n}{2 \sigma \|\mathbf d_{i,j}^{H} \|_F} \bigg]^2 \biggr)
		= \exp \bigg( - \frac{n^2 \delta_n^2}{8 \sigma^2\| \fd_{i,j}^{H} \|_F^2 } \bigg) .
	\end{align}
	Next, employing \eqref{Eq.Forb_d_UB} we obtain
	\begin{align}
		\label{Ineq.E_0}
		\Pr(\mathbbmss{E}_0) \leq \exp \bigg( - \frac{n^2 \delta_n^2}{8 \sigma^2\| \fd_{i,j}^{H} \|_F^2 } \bigg)
		\leq \exp \bigg( - \frac{n \delta_n^2}{32 \sigma^2 \sigma_{1}^2(\fH) P_{\max}^2 N_T} \bigg) \triangleq \zeta_0 .
	\end{align}

	Now, the complementary event $\mathbbmss{E}_0^c,$ gives
	\begin{align}
		\label{Eq.E_0_Comp}
		2 n^{-1} \big\langle \fZ,\, \bar{\fd}_{i,j}^{\fH} \big\rangle_F \geq - \delta_n .
	\end{align}
	Next, applying the law of total probability to the event $\mathbbmss{E}_2$ over $\mathbbmss{E}_0$ and its complement $\mathbbmss{E}_0^c,$ gives
	\begin{align}
		\Pr(\mathbbmss{E}_2) = \Pr(\mathbbmss{E}_2 \cap \mathbbmss{E}_0) + \Pr(\mathbbmss{E}_2 \cap \mathbbmss{E}_0^c)
		\stackrel{(a)}{\leq} \Pr(\mathbbmss{E}_0) + \Pr\left( \mathbbmss{E}_2 \cap\,{\mathbbmss{E}_0^c} \right)
		\stackrel{(b)}{\leq} \Pr\left( \mathbbmss{E}_0 \right) + \Pr\left( \mathbbmss{E}_1 \right), \hspace{-1mm}
		\label{Eq.TypeIIError-E_0+E_1} 
	\end{align}
	where $(a)$ uses $\mathbbmss{E}_2 \cap \mathbbmss{E}_0 \subset \mathbbmss{E}_0$ and $(b)$ holds by $\Pr(\mathbbmss{E}_2 \cap \mathbbmss{E}_0^c) \leq \Pr (\mathbbmss{E}_1 )$ which is proved in the following:
	\begin{align}
		\Pr(\mathbbmss{E}_2 \cap \mathbbmss{E}_0^c)
		& = \Pr\big( \big\{ W_{i,j} - N_R \sigma^2 \leq \delta_n \big\} \cap \big\{| V_{i,j} | \leq \delta_n \, \big\} \big)
		\nonumber\\&
		= \Pr \big( \big\{ U_{i,j} - N_R \sigma^2 \leq \delta_n - V_{i,j} \big\}
		\cap \big\{| V_{i,j} | \leq \delta_n \, \big\} \big)
		\nonumber\\&
		\overset{(a)}{\leq} \Pr \big( \big\{ U_{i,j} - N_R \sigma^2 \leq 2 \delta_n  \big\}  \big)
		\nonumber\\&
		= \Pr\left(\mathbbmss{E}_1 \right), \hspace{-1mm}
	\end{align}
	where $(a)$ holds since $\delta_n-V_{i,j} \leq 2\delta_n$ conditioned on $| V_{i,j} | \leq \delta_n.$
	
	We now proceed with bounding $\Pr\left(\mathbbmss{E}_1 \right)$ as follows. By codebook construction, every pair of distinct codewords is separated by at least $2\tau_n = 2\sqrt{n\epsilon_n}.$ Therefore, applying Lemma \ref{Lem.Min_dis_con_cb} implies
	\begin{align}
		\big\| \bar{\fd}_{i,j}^{\fH} \big\|_F
		= \|\fH(\bar{\fc}_i-\bar{\fc}_j)\|_F
		&\stackrel{\text{(a)}}{=}
		\|\fU_r\fSigma_r\fV_r^{\mathsf T}
		\fV_r(\fc_i-\fc_j)\|_F
		\nonumber\\
		&\stackrel{\text{(b)}}{=}
		\|\fU_r\fSigma_r(\fc_i-\fc_j)\|_F
		\nonumber\\
		&\stackrel{\text{(c)}}{=}
		\|\fSigma_r(\fc_i-\fc_j)\|_F
		\nonumber\\
		&\stackrel{\text{(d)}}{\geq}
		\sigma_r(\fH)\|\fc_i-\fc_j\|_F
		\nonumber\\
		&\stackrel{\text{(e)}}{\geq}
		2\sigma_r(\fH)\tau_n.
	\end{align}
	where,
	\begin{itemize}
		\item $(a)$ uses the reduced SVD
		$\fH=\fU_r\fSigma_r\fV_r^{\mathsf T}$ and use
		$\bar{\fc}_i=\fV_r\fc_i$ and
		$\bar{\fc}_j=\fV_r\fc_j$.
		\item $(b)$ exploits $\fV_r^{\mathsf T}\fV_r=\mathbf I_r$.
		\item $(c)$ holds since $\fU_r^{\mathsf T}\fU_r=\mathbf I_r$,
		left multiplication by $\fU_r$ preserves the Frobenius norm:
		\[
		\|\fU_r\fX\|_F^2
		=\operatorname{tr}\!\left(
		\fX^{\mathsf T}\fU_r^{\mathsf T}\fU_r\fX
		\right)
		=\|\fX\|_F^2.
		\]
		\item $(d)$ uses
		$\fSigma_r=\operatorname{diag}(\sigma_1,\ldots,\sigma_r)$
		and $\sigma_k\geq\sigma_r(\fH)$ for all $k$, we have
		$$
		\|\fSigma_r\fX\|_F^2
		=\sum_{k=1}^r\sigma_k^2\|\fX_{k,:}\|_2^2
		\geq
		\sigma_r^2(\fH)\|\fX\|_F^2.
		$$
		\item $(e)$ By the minimum-distance property of the codebook,
		$\|\fc_i-\fc_j\|_F\geq 2\tau_n$ for distinct $\fc_i$ and $\fc_j$.
	\end{itemize}
	Consequently,
	\begin{align}
		\label{Ineq.Radius_Event}
		\big\| \bar{\fd}_{i,j}^{\fH} \big\|_F^2
		&\geq 4\sigma_r^2(\fH)\tau_n^2.
	\end{align}
	
	Next, we obesrve that recalling \eqref{Eq.Epsilon_n} and the decoding threshold $\delta_n$ we obtain
	\begin{align}
		\label{Ineq.Radius_Event2}
		4 \sigma_r^2(\fH) \tau_n^2 = 4 \sigma_r^2(\fH) n \cdot \frac{a}{\sigma_r^2(\fH) n^{( 1 - b) / 2}}
		= \frac{4 n a}{n^{( 1 - b) / 2}} = 3 n \cdot \frac{4a}{3n^{(1 - b)/2}} = 3 n \delta_n .
	\end{align}
	Hence, merging \eqref{Ineq.Radius_Event} and \eqref{Ineq.Radius_Event2} we obtain
	\begin{align}
		\label{Ineq.Radius_Event3}
		- \big\| \bar{\fd}_{i,j}^{\fH} \big\|_F^2 \leq - 3 n \delta_n.
	\end{align}
	
	Therefore, merging \eqref{Eq.d_F_Decompos} and \eqref{Ineq.Radius_Event}, we can establish the following upper bound for the event $\mathbbmss{E}_1:$
	\begin{align}
		\label{Ineq.Event_E1}
		\Pr(\mathbbmss{E}_1)
		= \Pr \big( U_{i,j} - N_R \sigma^2 \leq 2\delta_n \big)
		& = \Pr \big( n^{-1} \big( \big\| \fZ \big\|_F^2
		+ \big\| \bar{\fd}_{i,j}^{\fH} \big\|_{F}^2 \big)
		- N_R \sigma^2 \leq 2\delta_n \big)
		\nonumber\\
		& = \Pr \big( n^{-1} \big\| \fZ \big\|_F^2
		- N_R \sigma^2
		\leq 2\delta_n - n^{-1}
		\big\| \bar{\fd}_{i,j}^{\fH} \big\|_{F}^2 \big)
		\nonumber\\
		& \stackrel{(a)}{\leq}
		\Pr \big( n^{-1} \big\| \fZ \big\|_F^2
		- N_R \sigma^2 \leq -\delta_n \big)
		\nonumber\\
		& \stackrel{(b)}{\leq}
		\Pr \big( \big| n^{-1} \big\| \fZ \big\|_F^2
		- N_R \sigma^2 \big| \geq \delta_n \big)
		\nonumber\\
		& \stackrel{(c)}{=}\Pr\big( \big| d_F(\fY(i),\bar{\fc}_i)
		- \mathbb{E}[d_F(\fY(i),\bar{\fc}_i)] \big|
		\geq \delta_n \big)
		\nonumber\\
		& \stackrel{(d)}{\leq}
		\frac{\operatorname{Var}[d_F(\fY(i),\bar{\fc}_i)]}{\delta_n^2}
		\nonumber\\
		& \stackrel{(e)}{=}
		\frac{\operatorname{Var}[n^{-1}\|\fZ\|_F^2]}{\delta_n^2}
		\stackrel{(f)}{\leq} \frac{2N_R\sigma^4}{n\delta_n^2}
		\triangleq \zeta_1 .
	\end{align}
	where
	\begin{itemize}
		\item $(a)$ follows from \eqref{Ineq.Radius_Event3}, which gives $-\|\bar{\fd}_{i,j}^{\fH}\|_F^2\leq-3n\delta_n$ and hence $2\delta_n-n^{-1}\|\bar{\fd}_{i,j}^{\fH}\|_F^2\leq-\delta_n$.
		\item $(b)$ follows from the event inclusion $\{X\leq-\delta_n\} \subseteq \{|X|\geq\delta_n\},$
		where $X\triangleq n^{-1}\|\fZ\|_F^2-N_R\sigma^2.$
		\item $(c)$ follows from the decoding metric definition and $\fY(i)=\mathbf H\bar{\fc}_i+\fZ$, which imply $d_F(\fY(i),\bar{\fc}_i)=n^{-1}\|\fZ\|_F^2-N_R\sigma^2$ and $\mathbb{E}[d_F(\fY(i),\bar{\fc}_i)]=0$.
		\item $(d)$ follows from the Chebyshev's inequality, which gives
		\[
		\Pr\big(
		|X-\mathbb{E}[X]|\geq\delta_n
		\big)
		\leq
		\frac{\operatorname{Var}[X]}{\delta_n^2},
		\qquad \delta_n>0.
		\]
		\item $(e)$ follows from the invariance of variance to constant shifts. Thus, $\operatorname{Var}[d_F(\fY(i),\bar{\fc}_i)] = \operatorname{Var}[n^{-1}\|\fZ\|_F^2]$.
		\item $(f)$ exploits $\operatorname{Var}[n^{-1}\|\fZ\|_F^2] = (2N_R\sigma^4)/n,$ cf. \eqref{Ineq.Var_Decoding_Metric}.
	\end{itemize}

	Therefore, employing the upper bounds given in \eqref{Ineq.E_0}, \eqref{Ineq.Event_E1} and \eqref{Eq.TypeIIError-E_0+E_1} yields
	\begin{align}
		P_{e,2}(i,j) \leq \Pr(\mathbbmss{E}_0) + \Pr(\mathbbmss{E}_1) \leq \zeta_0 + \zeta_1 \leq e_2,
	\end{align}
	hence, $P_{e,2}(i,j) \leq e_2$ holds for sufficiently large $n$ and arbitrarily small $e_2 > 0.$ We have thus shown that for every $e_1,e_2 > 0$ and sufficiently large $n$, there exists an $(n, M_n(R), e_1, e_2)$-DI code. This completes the achievability proof of Theorem~\ref{Th.MIMO-Capacity}.

	\section{Conclusions and Perspectives}
	\label{Sec.Conclusion}
	
	In this work, we provide a rigorous characterization of the identification
	problem over MIMO Gaussian channels, extending the classical SISO setting
	\cite{Salariseddigh_IT} to the multi-antenna case. We establish the achievability of identification codes with super-exponential codebook size \(M_n(R)=2^{(n\log n)R}\) and derive explicit lower and upper bounds on the achievable identification rate \(R\). Our results show that these bounds scale linearly with the MIMO channel rank \(r\). In particular, the achievable identification rate grows at least linearly with \(r\), enabling stronger super-exponential codebook growth as the channel rank increases. This characterizes the rank-dependent gain in achievable identification performance without requiring an exact multiplicative relationship between codebook sizes. These results characterize the fundamental identification capability of MIMO Gaussian channels and provide a basis for
	further investigation of multi-user and fading extensions. Nevertheless, developing a unified and fully general theory of DI for MIMO Gaussian channels, including non-asymptotic performance and imperfect channel knowledge, remains an important direction for future research. Further extensions to multi-user settings are of particular interest, including the characterization of identification rates for MIMO multiple-access channels in the uplink and MIMO broadcast channels in the downlink under various input constraints, such as average-power and peak-amplitude constraints. More broadly, extending the diversity–multiplexing framework of Zheng and Tse \cite{ZhengTse2003} to identification systems could shed light on the fundamental tradeoffs among identification performance, diversity, and multiplexing gains, and thereby establish new fundamental limits for MIMO identification.

	\subsection{Open Problems and Conjectures}
	Analytical evidence and derivations reported in \cite{Salari_ISI,Salari_Colored} suggest the following conjecture.
	\begin{conjecture}
		
		Let $\mathcal{G}$ denote a MIMO Gaussian channel with ISI and spatio-temporally correlated noise. We conjecture that the codebook-size scaling law $M_n(R)=2^{(n\log n)R}$ \emph{survives} in the presence of ISI and correlated Gaussian noise, provided that the corresponding channel spectrum is bounded away from zero. In particular, for noise covariance is given by the Kronecker prodcut $\mathbf{\Sigma}_z=\mathbf{R}_t\otimes\mathbf{R}_s,$ where $\mathbf{R}_s$ and $\mathbf{R}_t$ are the spatial and temporal covariance matrices, respectively, the same scaling is expected when decoding is performed using the corresponding generalized Mahalanobis metric.
		\[
		D_M(\mathbf{Y},\mathbf{H}\mathbf{X})
		\triangleq
		\Big\|
		\mathbf{R}_s^{-1/2}
		(\mathbf{Y}-\mathbf{H}\mathbf{X})
		\mathbf{R}_t^{-1/2}
		\Big\|_F.
		\]
	\end{conjecture}

	\section{Acknowledgments}
		
	The authors acknowledge Professor Dr. -Ing. Dr. rer. nat. Holger Boche for helpful discussions. Christian Deppe was supported by the BMFTR under the research program Communication Systems - “Souverän. Digital. Vernetzt.” through the joint project 6G-life (Project ID: 16KIS2415) and the project QSTARS (Project ID: 16KIS2196).

	\appendices

	\section{Proof of Lemma~\ref{Lem.Converse2}; Minimum Frobenius Distance}
	\label{App.Converse_Proof2}
	
	The proof relies on a key property of reliable identification codes: whenever the sum of the type-I and type-II error probabilities is strictly less than one, the output distributions induced by any two distinct codewords have total variation distance bounded away from zero. By the classical Fuchs--van de Graaf inequality, this yields a corresponding fidelity bound, which, together with the Gaussian fidelity expression, implies a lower bound on the Frobenius norm of the codewords difference.
	
	\begin{proof}
		Recall that in the converse a sequence of identification codes $\bigl(n, M(n,R), e_1^{(n)}, e_2^{(n)}\bigr),$ with codewords $\{\mathbf c_i\}_{i=1}^{M(n,R)}$ and decoding regions $\{\mathbbmss D_i\}_{i=1}^{M(n,R)}$. Fix any two distinct messages $i_1 \neq i_2$. We establish a lower bound on the total variation distance between $f_{\fZ}(\fy|\bar{\fc}_{i_1})$ and $f_{\fZ}(\fy|\bar{\fc}_{i_2}),$ i.e., $d_{\rm TV}(f_{\fZ}(\fy|\bar{\fc}_{i_1}) - f_{\fZ}(\fy|\bar{\fc}_{i_2}) ).$ Observe that induced by the definition of type I and type II errors,
		\begin{align}
			P_{e,1}(i_1) & = 1 -	\int_{\mathbbmss D_{i_1}}
			f_{\mathbf Z}(\mathbf y-\mathbf \fH \bar{\fc}_{i_1})\, d\mathbf y \le e_1^{(n)} \le e_1 .
			\nonumber\\
			P_{e,2}(i_2,i_1) & = \int_{\mathbbmss D_{i_1}}
			f_{\mathbf Z}(\mathbf y-\mathbf \fH \bar{\fc}_{i_2})\, d\mathbf y \le e_2^{(n)} \le e_2 .
		\end{align}
		Thereby, subtracting the output distributions over the decoding region $\mathbbmss D_{i_1}$ we obtain
		\begin{align}
			\label{Eq.Intg_Diff_Output}
			\int_{\mathbbmss D_{i_1}} \hspace{-2.7mm}
			\Bigl(
			f_{\mathbf Z}(\mathbf y- \fH \bar{\fc}_{i_1})
			-
			f_{\mathbf Z}(\mathbf y- \fH \bar{\fc}_{i_2})
			\Bigr)
			\, d\mathbf y
			\ge
			1-e_1^{(n)}-e_2^{(n)}.
		\end{align}
		Next, exploiting the supremum definition of total variation distance, we obtain
		\begin{align}
			d_{\rm TV}(f_{\fZ}(\fy|\bar{\fc}_{i_1}) - f_{\fZ}(\fy|\bar{\fc}_{i_2}) )
			& =
			\sup_{\mathcal A}
			\left|
			\int_{\mathcal A} (f_{\fZ}(\fy|\bar{\fc}_{i_1}) - f_{\fZ}(\fy|\bar{\fc}_{i_2}) \, d\mathbf y \right|
			\nonumber\\&
			\ge \bigg| \int_{\mathbbmss D_{i_1}}
			\Bigl( f_{\mathbf Z}(\mathbf y-\bar{\fc}_{i_1}) - f_{\mathbf Z}(\mathbf y-\bar{\fc}_{i_2})
			\Bigr) \, d\mathbf y \bigg| .
		\end{align}
		where the supremum is over all measurable sets $\mathcal A$ and $(b)$ holds by choosing the particular measurable set $\mathcal A=\mathbbmss D_{i_1}.$ Combining with \eqref{Eq.Intg_Diff_Output} yields
		\begin{align}
			\label{Ineq.LB_TV_D}
			d_{\rm TV}(f_{\fZ}(\fy|\bar{\fc}_{i_1}) - f_{\fZ}(\fy|\bar{\fc}_{i_2}) ) \ge 1 - e_1^{(n)} - e_2^{(n)}
			> 1 - e_1 - e_2 \triangleq e_0 .
		\end{align}
		Next, we leverage a refined approach, originating from quantum information theory, named, classical \emph{Fuchs–van de Graaf inequality} \cite{Wilde13} specialized to classical probability measures inequalities, i.e.,
		\begin{align}
			\label{Ineq.Fuchs_van_de_Graaf}
			1 - F(f_{\fZ}(\fy|\bar{\fc}_{i_1}),f_{\fZ}(\fy|\bar{\fc}_{i_2}))
			\le
			d_{\rm TV}(f_{\fZ}(\fy|\bar{\fc}_{i_1}),f_{\fZ}(\fy|\bar{\fc}_{i_2}))
			\le
			\sqrt{1 - F^2(f_{\fZ}(\fy|\bar{\fc}_{i_1}),f_{\fZ}(\fy|\bar{\fc}_{i_2}))} .
		\end{align}
		Now, merging \eqref{Ineq.LB_TV_D} and the right hand side of \eqref{Ineq.Fuchs_van_de_Graaf} gives
		\begin{align}
			F^2(f_{\fZ}(\fy|\bar{\fc}_{i_1}),f_{\fZ}(\fy|\bar{\fc}_{i_2})) \le 1 - e_0^2 \triangleq e_0' .
		\end{align}
		Taking the natural logarithm of both sides and then negating the result yields
		\begin{align}
			- \ln (e_0') & \le - \ln F^2(f_{\fZ}(\fy|\bar{\fc}_{i_1}),f_{\fZ}(\fy|\bar{\fc}_{i_2}))
			\nonumber\\&
			\stackrel{(a)}{\le} \frac{\|\fH(\bar{\fc}_{i_1}-\bar{\fc}_{i_2})\|_F^2}{4\sigma^2}
			\nonumber\\
			& \stackrel{(b)}{\leq}
			\frac{\| \fH \|_2^2 \cdot \| \bar{\fc}_{i_1}-\bar{\fc}_{i_2} \|_F^2}{4\sigma^2}
			\nonumber\\
			& \stackrel{(c)}{=}
			\frac{\sigma_{1}^2(\fH) \cdot \| \bar{\fc}_{i_1}-\bar{\fc}_{i_2} \|_F^2}{4\sigma^2} ,
		\end{align}
		where $(a)$ holds, as the fidelity between the two output distributions of $\bG$ is derived in Appendix~\ref{App.Fid_Gaussian}.
		\begin{align}
			F\!\left(
			f_{\fZ}(\fy|\bar{\fc}_{i_1}),
			f_{\fZ}(\fy|\bar{\fc}_{i_2})
			\right)
			=
			\exp\!\left(
			- \frac{\|\fH(\bar{\fc}_{i_1}-\bar{\fc}_{i_2})\|_F^2}{8\sigma^2}
			\right) ,
		\end{align}
		and $(b)$ exploits the spectral–Frobenius norm inequality, i.e., $\|\fA\fB\|_F \leq \|\fA\|_2 \|\fB\|_F.$ Thereby,
		\begin{align}
			\| \bar{\fc}_{i_1}-\bar{\fc}_{i_2} \|_F^2 \ge - \frac{4 \sigma^2 \ln (e_0')}{\sigma_{1}^2(\fH)} ,	
		\end{align}
		which implies
		\begin{equation}
			\| \bar{\fc}_{i_1}-\bar{\fc}_{i_2} \|_F \ge \frac{2 \sigma \sqrt{\ln (1/e_0')}}{\sigma_{1}(\fH)} .
		\end{equation}
		Next, we observe that multiplication by $\fV_r$ preserves the Frobenius norm, namely,
		\begin{align}
			\| \bar{\fc}_{i_1}-\bar{\fc}_{i_2} \|_F
			& \stackrel{\text{(a)}}{=}
			\|\fV_r(\fc_i-\fc_j)\|_F
			\nonumber\\
			& \stackrel{\text{(b)}}{=}
			\sqrt{\operatorname{tr}\big((\fV_r(\fc_i-\fc_j))^T
				\fV_r(\fc_i-\fc_j) \big)}
			\nonumber\\
			& =
			\sqrt{\operatorname{tr}\big((\fc_i-\fc_j)^T
				\fV_r^T \fV_r(\fc_i-\fc_j) \big)}
			\nonumber\\
			& \stackrel{\text{(c)}}{=}
			\sqrt{\operatorname{tr}\big((\fc_i-\fc_j)^T
				(\fc_i-\fc_j) \big)}
			\nonumber\\
			& = \| \fc_i - \fc_j \|_F ,
		\end{align}
		where
		\begin{itemize}
			\item $(a)$ uses the definition of the reduced codewords in \eqref{Eq.Precoder}.
			\item $(b)$ exploits the definition of the Frobenius norm in terms of the trace, i.e. $\|\fA\|_F^2 = \operatorname{tr}(\fA^{\mathsf T}\fA).$
			\item $(c)$ follows from the orthonormality of the columns of $\fV_r$, which gives
			$\fV_r^{\mathsf T}\fV_r=\fI_r$.
		\end{itemize}
		Thereby,
		\begin{align}
			\| \fc_{i_1} - \fc_{i_2} \|_F^2 \ge \frac{2 \sigma \sqrt{\ln (1/e_0')}}{\sigma_{1}(\fH)} .
		\end{align}
		This compelte the proof fo Lemma~\ref{Lem.Converse2}.
	\end{proof}

	\section{Fidelity Between Two Gaussian Distributions}
	
	\label{App.Fid_Gaussian}
	
	In the following, we derive the fidelity between the two Gaussian output
	distributions induced by the distinct codewords $\fc_{i_1}$ and $\fc_{i_2}$.
	For the channel model under consideration, the conditional output density
	given a transmitted codeword $\fx$ and channel matrix $\fH$ is given by \eqref{Eq.Noise_Den_Matr}
	\begin{align}
		\label{Eq:Conditional_Output_Density}
		f_{\fY|\fX}(\fy|\fx)
		& =
		\prod_{t=1}^{n}
		f_{\fY_t|\fX_t}(\fy_t|\fx_t)
		\nonumber\\&
		=
		\frac{1}{\sqrt{(2\pi\sigma^2)^{N_Rn}}}
		\exp\!\left(
		-\frac{\|\fy-\fH\fx\|_F^2}{2\sigma^2}
		\right).
	\end{align}
	For notational convenience, let
	\begin{equation}
		\label{Eq:Output_Densities_Definition}
		f_{\fZ}(\fy|\bar{\fc}_i)
		\triangleq
		f_{\fY|\fX}(\fy|\bar{\fc}_i).
	\end{equation}
	Thus, for the two transmitted codewords $\bar{\fc}_{i_1}$ and $\bar{\fc}_{i_2}$,
	\begin{equation}
		\label{Eq:Conditional_Density_1}
		f_{\fZ}(\fy|\fc_{i_1})
		=
		\frac{1}{\sqrt{(2\pi\sigma^2)^{N_Rn}}}
		\exp\!\left(
		-\frac{\|\fy-\fH\bar{\fc}_{i_1}\|_F^2}{2\sigma^2}
		\right),
	\end{equation}
	and
	\begin{equation}
		\label{Eq:Conditional_Density_2}
		f_{\fZ}(\fy|\bar{\fc}_{i_2})
		=
		\frac{1}{\sqrt{(2\pi\sigma^2)^{N_Rn}}}
		\exp\!\left(
		-\frac{\|\fy-\fH\bar{\fc}_{i_2}\|_F^2}{2\sigma^2}
		\right).
	\end{equation}
	The fidelity between these two output distributions reads
	\begin{equation}
		\label{Eq:Fidelity_Definition}
		F\!\left(
		f_{\fZ}(\fy|\bar{\fc}_{i_1}),
		f_{\fZ}(\fy|\bar{\fc}_{i_2})
		\right)
		=
		\int_{\mathbb{R}^{N_R \times n}} \hspace{-4mm}
		\sqrt{
			f_{\fZ}(\fy|\bar{\fc}_{i_1})
			f_{\fZ}(\fy|\bar{\fc}_{i_2})
		}
		\,d\fy.
	\end{equation}
	Substituting \eqref{Eq:Conditional_Density_1} and
	\eqref{Eq:Conditional_Density_2} into \eqref{Eq:Fidelity_Definition}
	gives
	\begin{equation}
		\label{Eq:Fidelity_Geometric_Mean}
		\sqrt{
			f_{\fZ}(\fy|\bar{\fc}_{i_1})
			f_{\fZ}(\fy|\bar{\fc}_{i_2})
		}
		=
		\frac{1}{\sqrt{(2\pi\sigma^2)^{N_R n}}}
		\exp\!\left(
		A_{i_1,i_2}(\fy)
		\right),
	\end{equation}
	where
	\begin{equation}
		\label{Eq:A_Definition}
		A_{i_1,i_2}(\fy)
		\triangleq
		-\frac{
			\|\fy-\fH\bar{\fc}_{i_1}\|_F^2
			+
			\|\fy-\fH\bar{\fc}_{i_2}\|_F^2
		}{4\sigma^2}.
	\end{equation}
	
	We now simplify the exponent. For any two real-valued vectors $\fa$ and
	$\fb$, we use the squared Euclidean norm identity
	\begin{equation}
		\label{Eq:Squared_Norm_Identity}
		\|\fa-\fb\|_F^2
		\triangleq
		(\fa-\fb)^T(\fa-\fb)
		=
		\fa^T\fa
		-
		2\fa^T\fb
		+
		\fb^T\fb.
	\end{equation}
	Applying \eqref{Eq:Squared_Norm_Identity} to the two terms appearing in
	\eqref{Eq:A_Definition} yields
	\begin{equation}
		\label{Eq:Norm_Expansion_1}
		\begin{aligned}
			\|\fy-\fH\bar{\fc}_{i_1}\|_F^2
			&=
			(\fy-\fH\bar{\fc}_{i_1})^T
			(\fy-\fH\bar{\fc}_{i_1})
			\\
			&=
			\fy^T\fy
			-
			2(\fH\bar{\fc}_{i_1})^T\fy
			+
			(\fH\bar{\fc}_{i_1})^T(\fH\bar{\fc}_{i_1}),
		\end{aligned}
	\end{equation}
	and
	\begin{equation}
		\label{Eq:Norm_Expansion_2}
		\begin{aligned}
			\|\fy-\fH\fc_{i_2}\|_F^2
			&=
			(\fy-\fH\bar{\fc}_{i_2})^T
			(\fy-\fH\bar{\fc}_{i_2})
			\\
			&=
			\fy^T\fy
			-
			2(\fH\bar{\fc}_{i_2})^T\fy
			+
			(\fH\bar{\fc}_{i_2})^T(\fH\bar{\fc}_{i_2}).
		\end{aligned}
	\end{equation}
	Adding \eqref{Eq:Norm_Expansion_1} and
	\eqref{Eq:Norm_Expansion_2} gives
	\begin{equation}
		\label{Eq:A_Expanded}
		\begin{aligned}
			-4\sigma^2 A_{i_1,i_2}(\fy)
			=
			2\fy^T\fy
			-
			2\big(
			\fH\bar{\fc}_{i_1}+\fH\bar{\fc}_{i_2}
			\big)^T\fy
			+
			\|\fH\bar{\fc}_{i_1}\|_F^2
			+
			\|\fH\bar{\fc}_{i_2}\|_F^2.
		\end{aligned}
	\end{equation}
	To complete the square, define the midpoint of the two received
	codewords by
	\begin{equation}
		\label{Eq:Received_Midpoint}
		\bar{\fm}_{i_1,i_2}^{\fH}
		\triangleq
		\frac{
			\fH\bar{\fc}_{ i_1}+\fH\bar{\fc}_{i_2}
		}{2}.
	\end{equation}
	Then
	\begin{equation}
		\label{Eq:Square_Completion}
		\begin{aligned}
			2\|\fy-\bar{\fm}_{i_1,i_2}^{\fH}\|_F^2
			=
			2\fy^T\fy
			-
			2\big(
			\fH\bar{\fc}v_{i_1}+\fH\bar{\fc}_{i_2}
			\big)^T\fy
			+
			2\|\bar{\fm}_{i_1,i_2}^{\fH}\|_F^2.
		\end{aligned}
	\end{equation}
	Consequently, \eqref{Eq:A_Expanded} becomes
	\begin{equation}
		\label{Eq:A_Completed_Square}
		\begin{aligned}
			-4\sigma^2 A_{i_1,i_2}(\fy)
			= 2\|\fy-\bar{\fm}_{i_1,i_2}^{\fH}\|_F^2
			+
			\|\fH\bar{\fc}_{i_1}\|_F^2
			+
			\|\fH\bar{\fc}_{i_2}\|_F^2
			-
			2\|\bar{\fm}_{i_1,i_2}^{\fH}\|_F^2.
		\end{aligned}
	\end{equation}
	It remains to simplify the term independent of $\fy$. From
	\eqref{Eq:Received_Midpoint},
	\begin{equation}
		\label{Eq:Midpoint_Constant}
		\begin{aligned}
			\|\fH\bar{\fc}_{i_1}\|_F^2
			+
			\|\fH\bar{\fc}_{i_2}\|_F^2
			-
			2\|\bar{\fm}_{i_1,i_2}^{\fH}\|_F^2
			&
			=
			\frac{1}{2}
			\left(
			\|\fH\bar{\fc}_{i_1}\|_F^2
			-
			2(\fH\bar{\fc}_{i_1})^T(\fH\bar{\fc}_{i_2})
			+
			\|\fH\bar{\fc}_{i_2}\|_F^2
			\right)
			\nonumber\\&
			= \frac{1}{2}
			\|\fH\bar{\fc}_{i_1}-\fH\bar{\fc}_{i_2}\|_F^2
			\nonumber\\&
			= \frac{1}{2}
			\|\fH(\bar{\fc}_{i_1}-\bar{\fc}_{i_2})\|_F^2.
		\end{aligned}
	\end{equation}
	Substituting \eqref{Eq:Midpoint_Constant} into
	\eqref{Eq:A_Completed_Square} and dividing by $-4\sigma^2$ gives the
	desired decomposition of the exponent:
	\begin{equation}
		\label{Eq:A_Final}
		A_{i_1,i_2}(\fy)
		=
		-\frac{
			\|\fy-\bar{\fm}_{i_1,i_2}^{\fH}\|_F^2
		}{2\sigma^2}
		-
		\frac{
			\|\fH(\bar{\fc}_{i_1}-\bar{\fc}_{i_2})\|_F^2
		}{8\sigma^2}.
	\end{equation}
	Substituting \eqref{Eq:A_Final} into
	\eqref{Eq:Fidelity_Geometric_Mean} and then into
	\eqref{Eq:Fidelity_Definition}, the fidelity separates into a
	codeword-dependent factor and a normalized Gaussian integral:
	\begin{equation}
		\label{Eq:Fidelity_Separated}
		\begin{aligned}
			& F\!\left(
			f_{\fZ}(\fy|\bar{\fc}_{i_1}),
			f_{\fZ}(\fy|\bar{\fc}_{i_2})
			\right)
			=
			\exp\!\left(
			-\frac{
				\|\fH(\bar{\fc}_{i_1}-\bar{\fc}_{i_2})\|_F^2
			}{8\sigma^2}
			\right)
			\\
			& \times
			\int_{\mathbb{R}^{N_R \times n}}
			\frac{1}{\sqrt{(2\pi\sigma^2)^{N_Rn}}}
			\exp\!\left(
			-\frac{
				\|\fy-\bar{\fm}_{i_1,i_2}^{\fH}\|_F^2
			}{2\sigma^2}
			\right)
			\,d\fy.
		\end{aligned}
	\end{equation}
	The integral in \eqref{Eq:Fidelity_Separated} is the integral of a
	normalized Gaussian density over its entire support and therefore equals
	one:
	\begin{equation}
		\label{Eq:Normalized_Gaussian_Integral}
		\int_{\mathbb{R}^{N_R \times n}}
		\frac{1}{\sqrt{(2\pi\sigma^2)^{N_Rn}}}
		\exp\!\left(
		-\frac{
			\|\fy-\bar{\fm}_{i_1,i_2}^{\fH}\|_F^2
		}{2\sigma^2}
		\right)
		\,d\fy
		=
		1.
	\end{equation}
	Hence, the fidelity between the two conditional output distributions is
	\begin{equation}
		\label{Eq:MIMO_Gaussian_Fidelity}
		F\!\left(
		f_{\fZ}(\fy|\bar{\fc}_{i_1}),
		f_{\fZ}(\fy|\bar{\fc}_{i_2})
		\right)
		=
		\exp\!\left(
		-\frac{
			\|\fH(\bar{\fc}_{i_1}-\bar{\fc}_{i_2})\|_F^2
		}{8\sigma^2}
		\right).
	\end{equation}
	Thus, the fidelity depends only on the squared Euclidean distance between the two codewords after propagation through the channel matrix $\fH$.

	\section{Proof of Lemma \ref{Lem.Min_dis_con_cb}; Channel-Domain Distance}
	\label{App.Min_dis_con_cb}
	
	\begin{proof}
		Recall that
		\begin{equation}
			\fH=\fU_r\fSigma_r\fV_r^{\mathsf T},
			\qquad
			\fU_r^{\mathsf T}\fU_r=\mathbf I_r,
			\qquad
			\fV_r^{\mathsf T}\fV_r=\mathbf I_r,
		\end{equation}
		where
		\[
		\fSigma_r=\operatorname{diag}(\sigma_1,\ldots,\sigma_r),
		\qquad
		\sigma_1\geq\cdots\geq\sigma_r>0.
		\]
		Since $\bar{\fc}_k=\fV_r\fc_k$, it holds $\bar{\fc}_i-\bar{\fc}_j=\fV_r(\fc_i-\fc_j).$
		Now, we observe that
		\begin{align}
			\|\fH(\bar{\fc}_i-\bar{\fc}_j)\|_F
			\nonumber&
			\stackrel{\text{(a)}}{=}
			\|\fU_r\fSigma_r\fV_r^{\mathsf T}
			\fV_r(\fc_i-\fc_j)\|_F \\
			\nonumber&
			\stackrel{\text{(b)}}{=}
			\|\fU_r\fSigma_r(\fc_i-\fc_j)\|_F \\
			\nonumber&
			\stackrel{\text{(c)}}{=}
			\|\fSigma_r(\fc_i-\fc_j)\|_F \\
			&\stackrel{\text{(d)}}{\geq}
			\sigma_r(\fH)\|\fc_i-\fc_j\|_F ,
		\end{align}
		where
		\begin{itemize}
			\item $(a)$ substitute the SVD $\fH=\fU_r\fSigma_r\fV_r^{\mathsf T}$.
			\item $(b)$ use $\fV_r^{\mathsf T}\fV_r=\mathbf I_r$.
			\item $(c)$ since $\fU_r^{\mathsf T}\fU_r=\mathbf I_r$, left multiplication by $\fU_r$
			preserves the Frobenius norm:
			\[
			\|\fU_r\fX\|_F^2
			=\operatorname{tr}(\fX^{\mathsf T}\fU_r^{\mathsf T}\fU_r\fX)
			=\|\fX\|_F^2.
			\]
			\item $(d)$ since $\fSigma_r=\operatorname{diag}(\sigma_1,\ldots,\sigma_r)$ with
			$\sigma_k\geq\sigma_r$,
			\[
			\|\fSigma_r\fX\|_F^2
			=\sum_{k=1}^r\sigma_k^2\|\fX_{k,:}\|_2^2
			\geq
			\sigma_r^2\|\fX\|_F^2.
			\]
		\end{itemize}
		
	\end{proof}

	\section{Proof of Lemma \ref{Lem.Vol_Preserv}; Volume Preservation under Vectorization}
	
	\label{App.Vol_Preserv}

	First, we verify that the vectorization map is linear. For
	$\fX,\fY\in\mathbb{R}^{r\times n}$ and $\alpha,\beta\in\mathbb{R}$,
	the column-wise definition of $\operatorname{vec}$ gives
	\begin{equation}
		\operatorname{vec}(\alpha\fX+\beta\fY)
		=
		\alpha\operatorname{vec}(\fX)
		+
		\beta\operatorname{vec}(\fY).
	\end{equation}
	Hence, $\operatorname{vec}$ is linear. Next, for $\fX,\fY\in\mathbb{R}^{r\times n}$,
	\begin{align}
		\left\langle
		\operatorname{vec}(\fX),
		\operatorname{vec}(\fY)
		\right\rangle_2
		=
		\sum_{j=1}^{n}\sum_{i=1}^{r}x_{ij}y_{ij} =
		\langle\fX,\fY\rangle_F.
	\end{align}
	Thus, $\operatorname{vec}$ preserves the inner product. In particular,
	by taking $\fY=\fX$,
	\begin{equation}
		\|\operatorname{vec}(\fX)\|_2 = \|\fX\|_F.
	\end{equation}	
	Moreover, by linearity,
	\begin{equation}
		\operatorname{vec}(\fX)-\operatorname{vec}(\fY) =
		\operatorname{vec}(\fX-\fY),
	\end{equation}
	and therefore
	\begin{align}
		\|\operatorname{vec}(\fX)-\operatorname{vec}(\fY)\|_2
		=
		\|\operatorname{vec}(\fX-\fY)\|_2 =
		\|\fX-\fY\|_F.
	\end{align}
	Hence, $\operatorname{vec}$ is a linear isometry from
	$\bigl(\mathbb{R}^{r\times n},\|\cdot\|_F\bigr)$ onto
	$\bigl(\mathbb{R}^{rn},\|\cdot\|_2\bigr)$. Moreover, with respect to the column-wise Frobenius-orthonormal basis
	
	\begin{equation}
		\operatorname{vec}(E)
		:=
		(E_{1j},\ldots,E_{rj})_{j=1}^{n}
	\end{equation}
	of $\mathbb{R}^{r\times n}$ and the standard basis of
	$\mathbb{R}^{rn}$, the vectorization map satisfies
	\begin{equation}
		\operatorname{vec}(E_{ij})
		=
		e_{(j-1)r+i}.
	\end{equation}
	Hence, its matrix representation is the identity matrix
	$I_{rn}$. Since $\operatorname{vec}$ is linear,
	\begin{equation}
		D\operatorname{vec}=I_{rn},
		\qquad
		\left|\det D\operatorname{vec}\right|=1.
	\end{equation}
	Therefore, by the change-of-variables formula, applied to
	$T=\operatorname{vec}$ and the function $g\equiv1$,
	\begin{align}
		\operatorname{Vol}_{2}\!\left(\operatorname{vec}(\mathcal A)\right)
		& =
		\int_{\operatorname{vec}(\mathcal A)}1\,d\fx
		\nonumber\\
		& =
		\int_{\mathcal A}
		\left|\det D\operatorname{vec}(\fX)\right|
		\,d\operatorname{Vol}_{F}(\fX)
		\nonumber\\
		& =
		\int_{\mathcal A}1\,d\operatorname{Vol}_{F}(\fX)
		\nonumber\\
		& =
		\operatorname{Vol}_{F}(\mathcal A).
	\end{align}
	
	Thus,
	\begin{equation}
		\operatorname{Vol}_{F}(\mathcal A)
		=
		\operatorname{Vol}_{2}\!\left(\operatorname{vec}(\mathcal A)\right).
	\end{equation}

	\section{Proof of Lemma \ref{Lem.Neg_Log_Vol}; Negative Log-Volume of Hypersphere}
	\label{App.Neg_Log_Vol}

	\begin{proof}
		Recall that
		\[
		\operatorname{Vol}\big(\S_{\fc_1}(n,r_0)\big)
		=
		\frac{(r_0\sqrt{\pi})^n}{\Gamma(1 + \frac n2)}.
		\]
		By Stirling's expansion \cite{Olver2010} for the Gamma function,
		\begin{align}
			\log\Gamma(x+1)
			&
			=
			\left(x+\frac12\right)\log x-x
			+\frac12\log(2\pi)
			+\frac{1}{12x}
			+O(x^{-3}),
		\end{align}
		for $x\to\infty.$ Setting \(x=n/2\), we obtain
		\[
		\log\Gamma\left(\frac n2+1\right)
		=
		\frac n2\log\frac n2-\frac n2
		+\frac12\log(\pi n)
		+\frac{1}{6n}
		+O(n^{-3}).
		\]
		Consequently,
		\[
		\begin{aligned}
			\log\operatorname{Vol}\big(\S_{\fc_1}(n,r_0)\big)
			&
			=
			n\log(r_0\sqrt{\pi})
			-\log\Gamma\left(\frac n2+1\right)\\
			&=
			n\log(r_0\sqrt{\pi})
			\hspace{-.5mm}-\frac n2\log\frac n2
			\hspace{-.5mm}+\frac n2
			\hspace{-.5mm}-\frac12\log(\pi n)
			\hspace{-.5mm}-\frac{1}{6n}
			\hspace{-.5mm}+O(n^{-3})\\
			& =
			\frac n2\log\left(\frac{2\pi e r_0^2}{n}\right)
			-\frac12\log(\pi n)
			+O(n^{-1}) .
		\end{aligned}
		\]
		Therefore,
		\begin{align}
			-\log\operatorname{Vol}\big(\S_{\fc_1}(n,r_0)\big)
			&
			=
			-\frac n2
			\log\left(
			\frac{2\pi e r_0^2}{n}
			\right)
			+\frac12\log(\pi n)
			+O(n^{-1})
			\\[2mm]
			&=
			\frac n2
			\log\left(
			\frac{n}{2\pi e r_0^2}
			\right)
			+\frac12\log(\pi n)
			+O(n^{-1})
			\\[2mm]
			&=
			\frac n2\log n
			-\frac n2\log 2
			-n\log r_0
			-\frac n2\log(\pi e)
			+o(n)
			\\[2mm]
			&=
			\frac n2\log\frac n2
			-n\log r_0
			-\frac n2\log(\pi e)
			+o(n)
		\end{align}
		since $$\frac12\log(\pi n)=o(n)
		\qquad\text{and}\qquad
		O(n^{-1})=o(n).$$
		This completes the proof of Lemma \ref{Lem.Neg_Log_Vol}.
		
	\end{proof}

	\section{Proof of Lemma \ref{Lem.Noise_Ch_Dis_Prod}; Noise-Channel-Distance Product}
	\label{App.Noise_Ch_Dis_Prod}
	
		\begin{proof}
		By the definition of the Frobenius inner product, we have
		\begin{equation}
			\label{Eq.V_ij}
			\begin{aligned}
				V_{i,j}
				&=
				2n^{-1}
				\big\langle
				\mathbf Z,\mathbf d_{i,j}^{H}
				\big\rangle_F
				& = 2n^{-1} \sum_{t=1}^{n}\sum_{\mu=1}^{N_R} Z_{t,\mu} \bar{d}_{i,j,t,\nu}^{H} .
			\end{aligned}
		\end{equation}
		
		We establish the distribution of $V_{i,j}$ in three steps:
		
		\begin{itemize}[leftmargin=*]
			\item $(a)$ \textit{Mean.}
			By linearity of expectation and the zero-mean property of the noise,
			\begin{equation}
				\begin{aligned}
					\mathbb E[V_{i,j}] = 2n^{-1} \sum_{t=1}^{n}\sum_{\mu=1}^{N_R} |\bar{d}_{i,j,t,\nu}^{H}|^{2} \mathbb E[Z_{t,\mu}]= 0.
				\end{aligned}
			\end{equation}
			
			\item $(b)$ \textit{Variance.}
			Since the noise samples $\{Z_{t,\mu}\}$ are independent with
			$\operatorname{Var}(Z_{t,\mu})=\sigma^2$, the variance of the sum is the sum of the variances. Thus,
			\begin{equation}
				\begin{aligned}
					\operatorname{Var}(V_{i,j})
					&=
					4n^{-2}
					\sum_{t=1}^{n}\sum_{\mu=1}^{N_R}
					|\bar{d}_{i,j,t,\nu}^{H}|^{2}
					\operatorname{Var}(Z_{t,\mu})\\
					&=
					4\sigma^2n^{-2}
					\sum_{t=1}^{n}\sum_{\mu=1}^{N_R}
					|\bar{d}_{i,j,t,\nu}^{H}|^{2} \\
					&=
					4\sigma^2n^{-2}
					\|\mathbf d_{i,j}^{H}\|_F^2.
				\end{aligned}
			\end{equation}
			
			\item $(c)$ \textit{Gaussianity.}
			Since, by \eqref{Eq.V_ij}, \(V_{i,j}\) is a deterministic linear combination of independent Gaussian RVs \(\{Z_{t,\mu}\}\), it is Gaussian.
		\end{itemize}
		Now, combining (a)--(c), we obtain $$V_{i,j} \sim \mathcal N\big( 0,,4\sigma^2n^{-2} \|\mathbf d_{i,j}^{H}\|_F^2 \big)$$ which completes the proof.
	\end{proof}

	\section*{}
	\bibliographystyle{IEEEtran-full}
	\bibliography{Lit}
	
\end{document}